\documentclass[aps,pra,reprint,footinbib,superscriptaddress]{revtex4-2}

\usepackage[utf8]{inputenc}  
\usepackage[T1]{fontenc}     
\usepackage[american]{babel}  
\usepackage[svgnames,dvipsnames]{xcolor}
\usepackage{hyperref}
\hypersetup{colorlinks=true, citecolor=FireBrick, urlcolor=blue, linkcolor=blue}
\usepackage{amsmath,amssymb,amsthm,bm,mathtools,amsfonts,mathrsfs,bbm,dsfont,times,appendix,multirow,color,cases} 
\usepackage{natbib}

\usepackage[ruled]{algorithm2e}

\usepackage{verbatim}
\newtheorem{lemma}{Lemma}
\newtheorem{observation}{Observation}
\newtheorem*{assumptions}{Assumptions}

\newcommand{\tr}{{\mathrm{tr}}}

\newcommand{\eins}{\mathbbm{1}}

\renewcommand{\vr}{\ensuremath{\varrho}}

\usepackage{braket}
\usepackage[normalem]{ulem}

\begin{document}
\title{Real randomized measurements for analyzing properties of quantum states}
\author{Jin-Min Liang}
\affiliation{State Key Laboratory for Mesoscopic Physics, School of Physics, Frontiers Science Center for Nano-optoelectronics, $\&$ Collaborative Innovation Center of Quantum Matter, Peking University, Beijing 100871, China}
\author{Satoya Imai}
\email{satoyaimai@yahoo.co.jp}
\affiliation{QSTAR, INO-CNR, and LENS, Largo Enrico Fermi, 2, 50125 Firenze, Italy}
\affiliation{Institute of Systems and Information Engineering, University of Tsukuba, Tsukuba, Ibaraki 305-8573, Japan}
\affiliation{Center for Artificial Intelligence Research (C-AIR), University of Tsukuba, Tsukuba, Ibaraki 305-8577, Japan}
\author{Shuheng Liu}
\affiliation{State Key Laboratory for Mesoscopic Physics, School of Physics, Frontiers Science Center for Nano-optoelectronics, $\&$ Collaborative Innovation Center of Quantum Matter, Peking University, Beijing 100871, China}
\author{Shao-Ming Fei}
\affiliation{School of Mathematical Sciences, Capital Normal University, Beijing 100048, China}
\author{Otfried G\"uhne}
\affiliation{Naturwissenschaftlich-Technische Fakult\"at, Universit\"at Siegen, Walter-Flex-Stra{\ss}e~3, 57068 Siegen, Germany}
\author{Qiongyi He}
\email{qiongyihe@pku.edu.cn}
\affiliation{State Key Laboratory for Mesoscopic Physics, School of Physics, Frontiers Science Center for Nano-optoelectronics, $\&$ Collaborative Innovation Center of Quantum Matter, Peking University, Beijing 100871, China}
\affiliation{Collaborative Innovation Center of Extreme Optics, Shanxi University, Taiyuan, Shanxi 030006, China}
\affiliation{Hefei National Laboratory, Hefei 230088, China}
\date{\today}
	
\begin{abstract}
    Randomized measurements are useful for analyzing quantum systems especially when quantum control is not fully perfect. However, their practical realization typically requires multiple rotations in the complex space due to the adoption of random unitaries. Here, we introduce two simplified randomized measurements that limit rotations in a subspace of the complex space. The first is \textit{real randomized measurements} (RRMs) with orthogonal evolution and real local observables. The second is \textit{partial real randomized measurements} (PRRMs) with orthogonal evolution and imaginary local observables. We show that these measurement protocols exhibit different abilities in capturing correlations of bipartite systems. We explore various applications of RRMs and PRRMs in different quantum information tasks such as characterizing high-dimensional entanglement, quantum imaginarity, and predicting properties of quantum states with classical shadow.
\end{abstract}
	
\maketitle
	
\section{Introduction}
Recent experimental progress in controlling and manipulating large quantum systems has enabled quantum advantages beyond classical regimes~\cite{arute2019quantum,wu2021strong,madsen2022quantum}. Extracting useful quantum information via measurements is essential for characterizing and analyzing large quantum systems. However, the curse of dimensionality poses challenges in reconstructing large density matrices by state tomography~\cite{mohseni2008quantum,flammia2012quantum}. Although several computationally and experimentally efficient tomography methods have been designed~\cite{torlai2018neural,quek2021adaptive,ahmed2021quantum,cramer2010efficient,toth2010permutationally,gross2010quantum}, these methods either still require complex optimization or are limited to specific states.
	
A feasible strategy is randomized measurements (RMs), which rotate measurement directions arbitrarily on each subsystem and collect measured data to analyze the properties of quantum states~\cite{elben2023randomized,cieslinski2024analysing}. Applications have emerged in various domains, including quantum simulation~\cite{kokail2021entanglement,notarnicola2023randomized}, state overlaps~\cite{elben2020cross,zhu2022cross,joshi2023exploring}, classical shadow tomography~\cite{huang2020predicting,akhtar2023scalable}, and the estimation of R{\'e}nyi entropies~\cite{elben2018renyi,brydges2019probing,elben2019statistical}, partial transpose moments~\cite{zhou2020single,elben2020mixed,neven2021symmetry,yu2021optimal,carrasco2024entanglement}, and quantum Fisher information~\cite{yu2021experimental,rath2021quantum,vitale2024estimation}.
	
In RMs, Alice and Bob apply local random unitaries $U_A \otimes U_B$ to a shared quantum state $\vr$, followed by local measurements $M_A \otimes M_B$. By sampling from the Haar measure on unitary group $\boldsymbol{\mathrm{U}}(d)$, they obtain the $t$-order moment as follows:
\begin{equation}\label{eq:RAB}
    R_{AB}^{(t)}\!=\!\! \int \! dU_A \! \! \int \! dU_B \,[\tr(\vr U_AM_AU_A^{\dag}\otimes U_BM_BU_B^{\dag})]^{t}.
\end{equation}
These moments provide an approach to characterize the quantum entanglement ~\cite{tran2015quantum,tran2016correlations,ketterer2019characterizing}. The techniques have been extended to high-dimensional systems~\cite{imai2021bound,liu2023characterizing,wyderka2023probing} and multiqubit systems~\cite{knips2020multipartite,ketterer2022statistically,wyderka2023complete}.

\begin{figure}[t]
    \centering
    \includegraphics[scale=0.3]{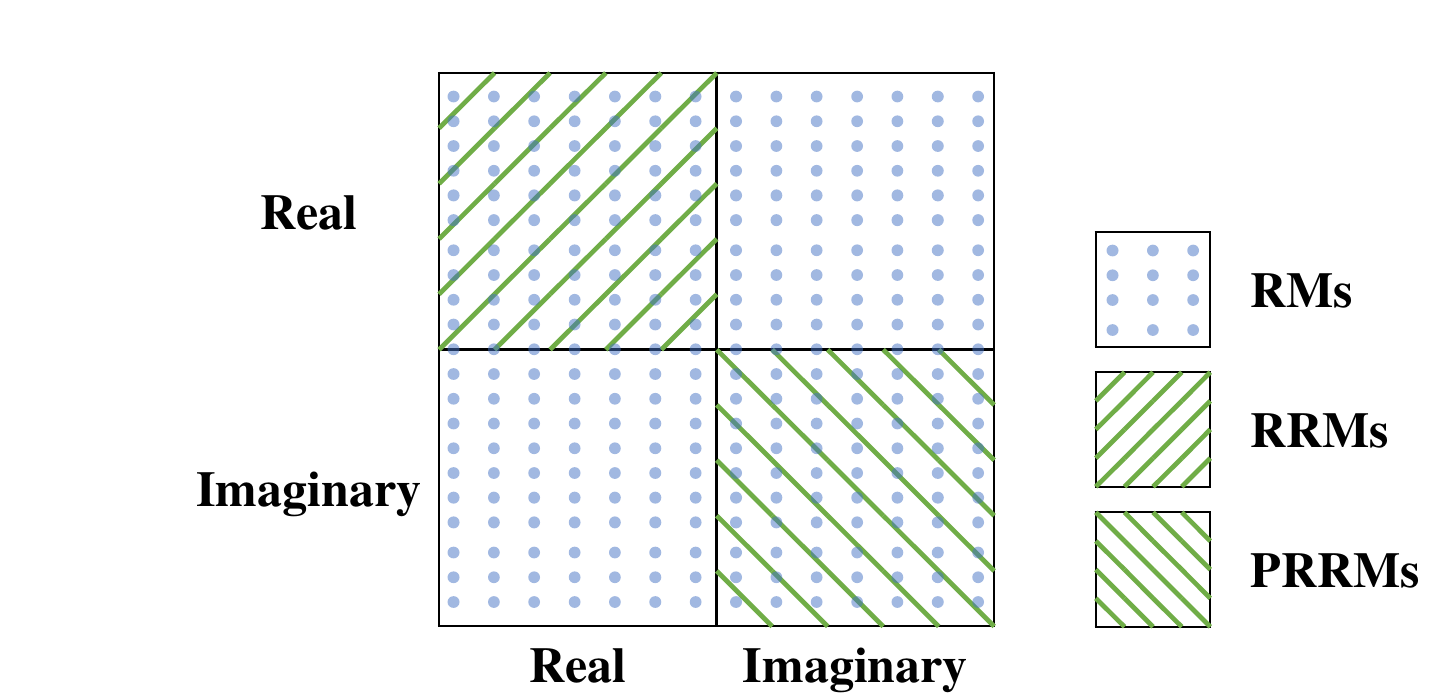}
    \caption{Layout for characterizing the correlations of bipartite quantum states from RMs, RRMs, and PRRMs, respectively, as given by Eqs.~(\ref{eq:RAB},\ref{RRMs},\ref{PRRMs}), formulated in Observation~\ref{Ob:correlation}. The horizontal and vertical axes correspond to one of the two subsystems each.}
    \label{Fig1}
\end{figure}

Although RMs are effective tools for analyzing quantum states, they still face several challenges. First, the practical implementation of RMs often requires multiple rotations in the whole complex space by successive tuning $d^2-1$ parameters for a $d$-dimension space. This may increase the susceptibility to environmental interactions and decoherence. Second, for frequently involved real quantum states~\cite{hardy2012limited,wootters2015optimal,ulyanov1992new,mckague2009simulating,wootters2012entanglement,aleksandrova2013real}, employing RMs designed for the complex full space as in Eq.~(\ref{eq:RAB}) may lead to unnecessary consumption of experimental resources. This insight was used for real randomized benchmarking to characterize the average gate fidelity over real-valued operators~\cite{hashagen2018real}. Finally, the moments in Eq.~(\ref{eq:RAB}) identify quantum correlations from both real and imaginary parts of quantum states, see Fig.~\ref{Fig1}, which hinders the analysis of the imaginarity of quantum states, known as a kind of quantum resource~\cite{hickey2018quantifying,wu2021operational,wu2021resource,kondra2023real,wu2024resource,fernandes2024unitary}.

In this paper, we generalize standard RMs to address these issues. In the following, we call a Hermitian operator $X$ real or imaginary if $X = X^\top$ or $X = - X^\top$, respectively, where $\top$ denotes the transposition. In our protocols, we consider assumptions of the following type:
\begin{assumptions}
    (i) The random unitaries of the form $U = e^{i\theta}O$ for $\theta \in [0, 2\pi)$ and orthogonal $O$ with $OO^{\top}=O^{\top}O = \eins$. (ii) The observables $M_A$ and $M_B$ are real. (iii) The observables $M_A$ and $M_B$ are imaginary.
\end{assumptions}
	
Here we call RMs \textit{real randomized measurements} (RRMs) if the RMs meet conditions (i) and (ii) in {Assumptions}, and \textit{partial real randomized measurements} (PRRMs) if the RMs meet (i) and (iii). First, we define the $t$-order moments of RRMs and PRRMs and discuss their relation with RMs. Specifically, RRMs and PRRMs extract real-real and imaginary-imaginary partial correlations of a bipartite state, while RMs extract all correlations as a whole. Then, we show that RRMs can characterize bipartite high-dimensional entanglement, including bound entanglement~\cite{peres1996separability,horodecki1997separability,horodecki1998mixed,bennett1999unextendible} and the Schmidt dimensionality~\cite{terhal2000schmidt,sanpera2001schmidt,krebs2024high}. We demonstrate that RRMs and RMs have equivalent entanglement detection strength for states with real-real correlations, but RRMs need fewer experimental steps than RMs. Next, combining RMs, RRMs, and PRRMs, we formulate a framework to detect the imaginarity of high-dimensional states and provide a lower bound for the existing imaginarity measure~\cite{hickey2018quantifying}. Our strategy does not require well-controlled measurement settings and maintains robustness against local orthogonal noise. Finally, we apply RRMs to classical shadow tomography~\cite{chen2021robust,nguyen2022optimizing}, requiring fewer random operations up to a fixed estimation precision than RMs.

\section{Real randomized measurements}
Suppose that Alice and Bob share a two-qudit state $\vr$. Following (i) in Assumptions, they apply orthogonal matrices $O_A$ and $O_B$ sampled via the Haar measure on the orthogonal group $\boldsymbol{\mathrm{O}}(d)$. The RRMs allow them to measure real observables $M_A$ and $M_B$, while the PRRMs allow them to measure imaginary observables $\hat{M}_A$ and $\hat{M}_B$. Here and hereafter, we use the notation $\hat{\square}$ to denote quantities associated with imaginarity. Similar to $t$-order moments of RMs~\cite{cieslinski2024analysing}, the $t$-order moments of RRMs and PRRMs are respectively defined as
\begin{align}
    Q_{AB}^{(t)}\!&=\!\! \int \! dO_A \! \! \int \! dO_B \,[\tr(\vr O_AM_AO_A^{\top}\otimes O_BM_BO_B^{\top})]^{t},\label{RRMs}\\
    \hat{Q}_{AB}^{(t)}\!&=\!\! \int \! dO_A \! \! \int \! dO_B \,[\tr(\vr O_A\hat{M}_AO_A^{\top}\otimes O_B\hat{M}_BO_B^{\top})]^{t}.
    \label{PRRMs}
\end{align}
	
Note that the moments $Q_{AB}^{(t)}$ and $\hat{Q}_{AB}^{(t)}$ are invariant under any local orthogonal operations, while the moment $R_{AB}^{(t)}$ is invariant under local unitary operations. That is, $Q_{AB}^{(t)}(\vr) = Q_{AB}^{(t)}(W_A \otimes W_B \vr W_A^{\top} \otimes W_B^{\top})$ for any $W_A, W_B \in \boldsymbol{\mathrm{O}}(d)$, and similarly for $\hat{Q}_{AB}^{(t)}$, while $R_{AB}^{(t)}(\vr) = R_{AB}^{(t)}(V_A \otimes V_B \vr V_A^{\dag} \otimes V_B^{\dag})$ for any $V_A, V_B \in \boldsymbol{\mathrm{U}}(d)$. In other words, the moments $Q_{AB}^{(t)}$ and $\hat{Q}_{AB}^{(t)}$ characterize two-qudit states in a real-reference-frame-independent manner~\cite{bartlett2007reference,klockl2015characterizing}.
	
Let us discuss the difference among RMs, RRMs, and RRMs. Recall that any two-qudit state $\vr$ can be written as
\begin{equation}\label{TwoQudit}
    \vr=\frac{1}{d^2}\sum_{j,k=0}^{d^2-1}T_{jk}\lambda_j\otimes\lambda_k,
\end{equation}
where $\lambda_0=\eins_d$ denotes the $d$-dimension identity and $\lambda_j$ are the generalized Gell-Mann matrices (GGMs) satisfying $\lambda_{j}=\lambda_{j}^{\dag}$, $\tr(\lambda_{j}\lambda_{k})=d\delta_{jk}$, and $\tr(\lambda_{j})=0$ for $j>0$ \cite{kimura2003bloch,bertlmann2008bloch}. Note that each of $(d^2-1)$ GGMs can be seen as an observable, belonging to one of the two types: (a) \textit{real} GGMs consisting of $L$ real observables: $\lambda_j^\top = \lambda_j$, associating with condition (ii) in Assumptions; or (b) \textit{imaginary} GGMs consisting of $\hat{L}$ imaginary observables: $\lambda_j^\top = - \lambda_j$, associating with condition (iii) in Assumptions. Here, we let first $L=(d-1)(d+2)/2$ GGMs be real GGMs and latter $\hat{L}=d(d-1)/2$ GGMs be imaginary GGMs, where $L + \hat{L} = d^2-1$, see Appendix~\ref{A_Preliminaries}.
	
The bipartite correlations in $\vr$ can then be divided into four types: real-real, imaginary-imaginary, real-imaginary, and imaginary-real ones, as shown in Fig.~\ref{Fig1}. We have the following observation.
	
\begin{observation}\label{Ob:correlation}
    The RMs given by Eq.(\ref{eq:RAB}) extract all four types of correlations of the $\vr$ as a whole, while RRMs and PRRMs given by Eqs.(\ref{RRMs}) and (\ref{PRRMs}) only extract the real-real and the imaginary-imaginary correlations of $\vr$, respectively, as shown in Fig.~\ref{Fig1}. The following relation holds for all $\vr$,
    \begin{equation}\label{ob1:eq1}
        Q_{AB}^{(t)}(\vr) + \hat{Q}_{AB}^{(t)}(\vr) \leq R_{AB}^{(t)}(\vr),\quad\forall \vr,
    \end{equation}
    where the equality is attained when $\vr$ is real, $\vr = \vr^{\top}$.
\end{observation}

To prove Observation~\ref{Ob:correlation}, we consider the simplest case of $t=2$. General cases for any $t$ are given in Appendix~\ref{B_Proofs}. Based on the Schur–Weyl duality in unitary and orthogonal groups (see Refs.~\cite{goodman2000representations,audenaert2002asymptotic,collins2006integration,mele2024introduction,garcia2025quantum} or Appendix~\ref{A_Preliminaries}), the second commutant in the unitary group is spanned by identity $\eins_{d}\otimes\eins_{d}$ and the SWAP operator $\mathbb{S} =\sum_{j,k}|jk\rangle\langle kj|$. In contrast, the one in the orthogonal group is spanned by the identity, the SWAP operator, and additionally, the operator $\Pi=\sum_{j,k}|jj\rangle\langle kk|$, where $\{|j\rangle\}$ denotes the computational basis in $d$ dimensions. Hence, the Haar unitary and orthogonal integrals read, respectively,
\begin{align}
    &\int dU \, U^{\otimes2}\mathcal{A}U^{\dag \otimes2}=\alpha \eins_{d}\otimes\eins_{d}+ \beta \mathbb{S},\label{eq:unitarysecond}\\
    &\int dO \, O^{\otimes2}\mathcal{A}O^{\top\otimes2}=\gamma_{1}\eins_{d}\otimes\eins_{d}+\gamma_{2}\mathbb{S}+\gamma_{3}\Pi,\label{eq:orthogonalsecond}
\end{align}
where the parameters $\alpha, \beta, \gamma_{1}, \gamma_{2}$, and $\gamma_{3}$ depend on a two-qudit operator $\mathcal{A}$, see Appendix~\ref{A_Preliminaries}. Thus, the three second-order moments of RRMs, PRRMs, and RMs are, respectively, $Q_{AB}^{(2)}(\vr):=\sum_{j,k=1}^{L}T_{jk}^2$, $\hat{Q}_{AB}^{(2)}(\vr):=\sum_{j,k=L+1}^{d^2-1}T_{jk}^2$, and $R_{AB}^{(2)}(\vr):=\sum_{j,k=1}^{d^2-1}T_{jk}^2$ by ignoring normalization constants. Note that the real-imaginary correlations $\sum_{j=1}^{L}\sum_{k=L+1}^{d^2-1}T_{jk}^2$ and imaginary-real correlations $\sum_{j=L+1}^{d^2-1}\sum_{k=1}^{L}T_{jk}^2$ are missing. Thus, Eq.(\ref{ob1:eq1}) is true for any $\vr$. Moreover, real states (i.e. $\vr=\vr^{\top}$) exhibit real-real and imaginary-imaginary correlations, which guarantees that equality in Eq.(\ref{ob1:eq1}) holds for real states.

From Observation~\ref{Ob:correlation}, RRMs and PRRMs can extract and distinguish partial correlations. This feature determines their irreplaceable roles in quantum information tasks. We show in the following that real-real correlations are still useful for entanglement detection.

\section{Entanglement detection via RRMs}
We here employ RRMs to characterize high-dimensional bipartite entanglement. Recall that the Schmidt number (\textrm{SN}) of mixed state $\vr$ is defined by the Schmidt rank (\textrm{SR}) of pure state $|\psi_j\rangle$~\cite{terhal2000schmidt}, $\textrm{SN}(\vr)=\min\{r:\vr=\sum_{j}p_j|\psi_j\rangle\langle\psi_j|,\textrm{SR}(|\psi_j\rangle)\leq r\}$, where $p_j\geq0$ and $\sum_{j}p_j=1$. We aim to find an analytical criterion to detect \textrm{SN}$(\vr)$ using RRMs.
	
The first step is to evaluate the orthogonal integrals of the moments $Q_{AB}^{(t)}(\vr)$ in Eq.~(\ref{RRMs}) and find their analytical expressions. For $t=2,4$, we obtain
\begin{align}
    \label{eq:secondmoment}
    Q_{AB}^{(2)} &= \frac{1}{L^2} \sum_{j=1}^{L}\tau_j^2,\\\label{eq:fourthmoment}
    Q_{AB}^{(4)} &=W\left[2\sum_{j=1}^{L}\tau_j^4+L^4\left(Q_{AB}^{(2)}\right)^2\right],
\end{align}
where $\tau_j$ are singular values of the submatrix $T_R=(T_{jk})$ of the correlation matrix with the entries given in Eq.~(\ref{TwoQudit}); $j,k >0$ running over all real GGMs; $L=(d-1)(d+2)/2$ is the number of real GGMs; and the normalization $W=\frac{3\Gamma^2(\frac{L}{2})}{16\Gamma^2\left(\frac{L+4}{2}\right)}$, $\Gamma(z)$ denotes Euler's gamma function with $\Gamma(z)=(z-1)!$ for positive integer $z$~\cite{folland2001integrate}. See details in Appendix~\ref{C_Estimating}.
	
It is important to note that for higher-order moments, not all observables give rise to the same results, since orthogonal matrices do not imply sphere rotations, due to the lack of the Bloch sphere for $d>2$. Nevertheless, one can construct observables to obtain the moments in Eqs.~(\ref{eq:secondmoment} and \ref{eq:fourthmoment}) for $t=2,4$ and any dimension $d$~\cite{imai2021bound,wyderka2023probing}, for details see Appendix~\ref{C_Estimating}. For example, the observables for $d=3$ are $M_A=M_B=\mathrm{diag}(\sqrt{3/2},0,-\sqrt{3/2})$~\cite{wyderka2023probing}.

Next, we propose our entanglement dimensionality criterion by adopting the second and fourth moments of RRMs. Our approach to detect \textrm{SN}$(\vr)$ is to consider a plane spanned by $(Q_{AB}^{(2)},Q_{AB}^{(4)})$ defined in Eqs.~(\ref{eq:secondmoment} and \ref{eq:fourthmoment}) by employing the strategies in Refs~\cite{ketterer2019characterizing,imai2021bound,wyderka2023probing,liu2023characterizing}. The task is to obtain the lower boundary on this plane, so as to reveal a hierarchy of Schmidt numbers. To do so, we use the criterion~\cite{wyderka2023probing,liu2023characterizing} that any bipartite quantum state $\vr$ with $\textrm{SN}(\vr)\leq r$ satisfies $\|T\|_{\tr}\leq rd-1$ for $r=1,\cdots,d$, where $\|\cdot\|_{\tr}$ denotes the trace norm, i.e., the sum of the singular values of $T$ as in Eq.~(\ref{TwoQudit}). In the case of $r=1$, the inequality reduces to the de Vicente criterion~\cite{de2006separability}. Recalling that $T_R$ is the submatrix of $T$, we have that $\|T_\mathrm{R}\|_{\tr}=\sum_{j=1}^{L}\tau_j\leq\|T\|_{\tr}\leq rd-1$~\cite{thompson1972principal}. We summarize our criteria as follows:
\begin{observation}\label{Ob:EnDetect}
    Let $C_2$ and $C_4$ be the second and fourth moments of RRMs given in Eqs.~(\ref{eq:secondmoment} and \ref{eq:fourthmoment}) for a two-qudit state $\vr$. Consider the optimization problem
    \begin{align}
        \min_{\tau_j}\quad & Q_{AB}^{(4)},\quad
        \mathrm{s.t.}\quad 
		\begin{cases}
		  Q_{AB}^{(2)}=C_2,\\
		  \sum_{j=1}^{L}\tau_j\leq rd-1,\\
		  \tau_j\in[0,d-1].
		\end{cases}
    \end{align}
    If $F_{\min}(r,C_2)\leq C_4<F_{\min}(r-1,C_2)$, then $\textrm{SN}(\vr)=r$ for $r=2,\cdots,d$. Here, we denote the minimum of the above optimization problem as
    \begin{equation}
	F_{\min}\left(x,y\right)\!=\!
	\begin{cases}
		\! WL^3y^2(2+L), & 0 \!\leq\! y \!\leq\! \frac{(dx-1)^2}{L^3},\\
		\!F_{\min}^{(n_{g})}\left(x,y\right), & \frac{(dx-1)^2}{L^2\left(n_{g}+1\right)} \!\leq\! y\!<\!\frac{(dx-1)^2}{L^2n_{g}},
	\end{cases}\nonumber
\end{equation}
where $x=1,\cdots,d$, $n_g=1,\cdots,L-1$, $F_{\min}^{(n_{g})}\left(x,y\right)=2Wf\left(n_{g},y\right)/\left(n_{g}+1\right)^4+WL^4y^2$, $f\left(n_{g},y\right)=\left[\mathcal{B}-(dx-1)\right]^4+\left[\mathcal{B}+n_{g}(dx-1)\right]^4/n_{g}^3$, and $\mathcal{B}=\sqrt{n_{g}\left(n_{g}+1\right)L^2y-n_{g}\left(dx-1\right)^2}$.
\end{observation}

\begin{figure}[t]
    \centering
    \includegraphics[scale=0.5]{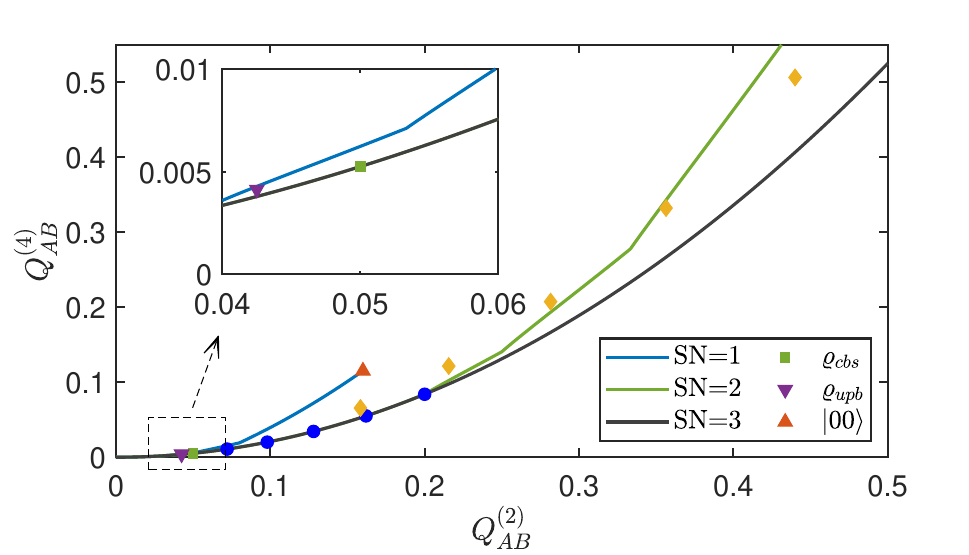}
    \caption{Entanglement dimensionality criterion based on second and fourth moments of RRMs for two-qutrit systems, formulated in Observation~\ref{Ob:EnDetect}. The blue dots from right to left denote the two-qutrit isotropic state $\vr_{\mathrm{iso}}$ with $p=1,0.9,0.8,0.7,0.6$. The orange diamonds from right to left represent the two-qutrit state $\vr(u)=u\vr_0+(1-u)\eins^{\otimes2}/3^2$ with $u=1,0.9,0.8,0.7,0.6$, where the state $\vr_0=(\eins^{\otimes2}+2\lambda_1\otimes\lambda_1+\lambda_1\otimes\lambda_3+\lambda_3\otimes\lambda_3+\lambda_4\otimes\lambda_4+2\lambda_6\otimes\lambda_6)/3^2$. The bound entangled $\vr_{upb}$ and $\vr_{cbs}$ are outside the separable area meaning that our criterion detects these states.} 
    \label{Fig2}
\end{figure}

The proof of Observation~\ref{Ob:EnDetect} is given in Appendix~\ref{B_Proofs}. From Observation~\ref{Ob:EnDetect}, we can reveal the entanglement and Schmidt dimensionality from the measured moments $(Q_{AB}^{(2)}, Q_{AB}^{(4)})$. To see this, Fig.~\ref{Fig2} shows the detection results of our method for the two-qutrit isotropic state $\vr_{\mathrm{iso}}=p|\Phi_{3}^{+}\rangle\langle\Phi_{3}^{+}|+(1-p)\eins^{\otimes2}/3^2$ for $p\geq0.6$ and the parametrized state $\vr(u)=u\vr_0+(1-u)\eins^{\otimes2}/3^2$ for $u\geq0.6$, where $|\Phi_{3}^{+}\rangle=({1}/{\sqrt{3}})\sum_{j=0}^{2}|jj\rangle$ is the maximally entangled state and $\vr_0=(\eins^{\otimes2}+2\lambda_1\otimes\lambda_1+\lambda_1\otimes\lambda_3+\lambda_3\otimes\lambda_3+\lambda_4\otimes\lambda_4+2\lambda_6\otimes\lambda_6)/3^2$. Bound entangled states with positive under partial transposition (PPT)~\cite{peres1996separability,horodecki1997separability,horodecki1998mixed,bennett1999unextendible} exhibit weak entanglement and are usually difficult to detect. Surprisingly, Observation~\ref{Ob:EnDetect} for $r=1$ can detect the bound entangled states such as the chessboard state $\vr_{cbs}$~\cite{bruss2000construction} and the unextendible product base state $\vr_{upb}$~\cite{bennett1999unextendible}, as displayed in Fig.~\ref{Fig2}.

We remark that entanglement detection approaches based on RRMs and RMs~\cite{imai2021bound} share a similar framework but employ different random operations. Since RMs extract more correlations than RRMs as formulated in Observation~\ref{Ob:correlation}, RMs can naturally detect more entangled states than RRMs. For example, the bound entangled Horodecki state~\cite{horodecki1998mixed} is detectable by RMs~\cite{imai2021bound}, but not by RRMs. Note that RMs and RRMs have an equivalent detection strength for the state that only has real-real correlations since both RMs and RRMs can extract their total correlations. As a subset of the PPT states, the partial transpose invariant (PTI) state~\cite{kraus2000separability,huber2018high}, $\vr_{AB}=\vr_{AB}^{\top_A}$, has only real-real correlations. It is easy to check that the chessboard state $\vr_{cbs}$~\cite{bruss2000construction} and the unextendible product base state $\vr_{upb}$~\cite{bennett1999unextendible} are the PTI states. This implies that our criteria can detect bound entangled states $\vr_{cbs}$ and $\vr_{upb}$. However, in this scenario, RRMs need fewer experimental operations than RMs. For linear optics, implementing a $d\times d$ unitary requires at least $d^2-1$ unset wave plates, while $(d^2-d)/2$ unset wave plates are required for an orthogonal matrix~\cite{wu2021resource,wu2024resource}. For large $d$, RRMs reduce the number of unset wave plates by $1/2$ compared to RMs, providing a more economical way to analyze high-dimensional entangled states.

\section{Detection and quantification of imaginarity}
RMs are currently used to analyze entanglement~\cite{cieslinski2024analysing,elben2023randomized} and magic resources~\cite{oliviero2022measuring}. Here, we show that RMs together with RRMs and PRRMs can be used to characterize imaginarity~\cite{hickey2018quantifying,wu2021operational,wu2021resource}, another important resource in channel discrimination~\cite{wu2021strong,wu2024resource}, weak-value theory~\cite{kedem2012using}, and quantum multiparameter metrology~\cite{miyazaki2022imaginarity}. From the imaginarity resource theory~\cite{hickey2018quantifying,wu2021operational,wu2021resource}, the resource states are identified as imaginary states that have imaginary entries on a given basis $\{|j\rangle\}$. Our task is to detect and quantify imaginary states without a common reference frame. However, this is impossible if one only uses the moments $R_{AB}^{(t)}$ in Eq.(\ref{eq:RAB}), as illustrated in Fig.~\ref{Fig1}. To proceed, we consider the combination of RRMs and PRRMs, and summarize our results below:
\begin{observation}\label{Ob:ImagTwo}
    Consider the second moments $R_{AB}^{(2)}$ of RMs, $Q_{AB}^{(2)}$ of RRMs, and $\hat{Q}_{AB}^{(2)}$ of PRRMs for a two-qudit state $\vr_{AB}$. Define the gap
    \begin{equation}
        G_{AB}^{(2)}=(d^2-1)^2R_{AB}^{(2)}-L^2Q_{AB}^{(2)}-\hat{L}^2\hat{Q}_{AB}^{(2)},
    \end{equation}
    where $L=(d-1)(d+2)/2$, $\hat{L}=d(d-1)/2$. Let $\hat{Q}_X^{(2)}$ be the marginal second moments of the reduced states $\vr_X$ $(X=A,B)$. We have the following:
		
    (i) $\vr_{AB}$ is a real (free) state if and only if the following conditions hold: $\hat{Q}_{A}^{(2)} = \hat{Q}_{B}^{(2)} = G_{AB}^{(2)} = 0$. Conversely, $\vr_{AB}$ is an imaginary (resource) state if and only if at least one of the following conditions holds: (a) $\hat{Q}_{A}^{(2)}>0$, (b) $\hat{Q}_{B}^{(2)}>0$, (c) $G_{AB}^{(2)}>0$.    
		
    (ii) $\hat{Q}_{A}^{(2)}$, $\hat{Q}_{B}^{(2)}$, and $G_{AB}^{(2)}$ induce a lower bound of the robustness of the imaginarity measure $\mathcal{F}_{R}(\vr_{AB})=\min_{\tau}\{\mu\geq0: \frac{\vr_{AB}+\mu\tau}{1+\mu}\in\mathscr{R}\}$~\cite{wu2021operational}, where $\mathscr{R}$ denotes the set of all two-qudit real states and the minimum is taken over all quantum states $\tau$,
    \begin{equation}
        \mathcal{F}_{\mathrm{LB}}(\vr_{AB}) \leq \mathcal{F}_{R}(\vr_{AB}),
    \end{equation}
    where $\mathcal{F}_{\mathrm{LB}}=\frac{1}{d}\sqrt{\hat{L}\left(\hat{Q}_A^{(2)}+\hat{Q}_B^{(2)}\right)+G_{AB}^{(2)}}$.
\end{observation}
	
The proof of Observation~\ref{Ob:ImagTwo} is mainly based on the relation of Eq.(\ref{ob1:eq1}) for $t=2$. Conditions (a) and (b) identify the imaginarity of the reduced states, while condition (c) implies the imaginarity of bipartite correlation. All three RMs, RRMs, and PRRMs are required to fully characterize the imaginarity of two-qudit states. On the other hand, to detect the imaginary of single-qudit states, we only need PRRMs or a combination of RMs and RRMs. This follows from the fact that the second moment of PRRMs $\hat{Q}_X^{(2)}$ coincides with the single-qudit gap $G_{X}=(d^2-1)R_{X}^{(2)}-LQ_{X}^{(2)}$ based on RMs and RRMs. Table~\ref{ImagResults} shows examples of three two-qutrit states. In Appendix~\ref{E_Generalized}, observation~\ref{Ob:ImagTwo} can be generalized to multipartite states.

\begin{table}
    \caption{Two-qutrit states: (1) $(|00\rangle+i|22\rangle)/\sqrt{2}$, (2) $(i|02\rangle+i|12\rangle+|10\rangle+|12\rangle)/2$, and (3) $(|00\rangle+|22\rangle)/\sqrt{2}$.}
    \begin{tabular}{c|c|c|c|c|c|c}
        \hline\hline
        $\vr_{AB}$ & $\hat{Q}_A^{(2)}$ & $\hat{Q}_B^{(2)}$ & $G_{AB}^{(2)}$ & $\mathcal{F}_{\mathrm{LB}}(\vr_{AB})$ & $\mathcal{F}_{R}(\vr_{AB})$ & Result\\
        \hline
        (1) & $0$ & $0$ & $4.5$ & $0.7071$ & $1$ & Imaginary\\
        \hline
        (2) & $0.1250$ & $0.1250$ & $2.6250$ & $0.6124$ & $0.8660$ & Imaginary\\
        \hline
        (3) & $0$ & $0$ & $0$ & $0$ & $0$ & Real\\
        \hline\hline
    \end{tabular}
    \label{ImagResults}
\end{table}

The traditional characterization methods for measuring imaginarity involving complicated optimization and quantum state tomography are fragile to noise due to the requirement for well-controlled measurement settings~\cite{hickey2018quantifying}. Notably, our strategy performs random measurements with a real-reference-frame-independent manner and is naturally robust to local orthogonal noise. This result motivates experimental approaches to analyze the imaginarity of quantum theory.

\section{Classical shadows with RRMs}
The classical shadow is a powerful tool in certifying properties of quantum systems \cite{huang2020predicting,chen2021robust,nguyen2022optimizing}. The standard classical shadow is based on RMs by employing random unitary and computational basis measurements.
	
Next, we develop the classical shadow with RRMs following (i) and (ii) in Assumptions. Applying orthogonal operation $O$ for $\vr\rightarrow O\vr O^{\top}$ and performing projective measurements in the basis $\{|b\rangle\}$, we store the classical description of $O^{\top}|b\rangle\langle b|O$. The overall procedure is associated with the quantum channel $\mathcal{M}(\rho)=\mathbb{E}\left[O^{\top}|b\rangle\langle b|O\right]$. The learned classical representation (classical shadows) is then given by $\tilde{\vr}=\mathcal{M}^{-1}\left(O^{\top}|b\rangle\langle b|O\right)$, where $\mathcal{M}^{-1}$ is the inverse of $\mathcal{M}(\vr)$. We derive the following result:
\begin{observation}\label{Ob:CS}
    By performing a global orthogonal matrix $O$, we have for an $N$-qudit real state $\vr$:
    \begin{equation}\label{MC}
        \mathcal{M}(\vr)=\frac{\eins_{d^N}+2\vr}{d^N+2},\quad
        \mathcal{M}^{-1}(\vr)=\frac{d^N+2}{2}\vr-\frac{1}{2}\eins_{d^N}.
    \end{equation}
\end{observation}
	
The proof of Observation~\ref{Ob:CS} is given in Appendix~\ref{B_Proofs}, where we use the second commutant in the orthogonal group~\cite{goodman2000representations,audenaert2002asymptotic,collins2006integration,mele2024introduction,garcia2025quantum}. Employing $N_O$ global orthogonal matrices, we obtain a dataset of classical shadows, $\{\tilde{\vr}^{(m)}\}_{m=1}^{N_O}$, which can be used to predict the linear and nonlinear functions of $\vr$. For local random orthogonal matrix $O=\bigotimes_{j=1}^{N}O_j$, the classical shadows of $\vr$ are given by $\tilde{\vr}=\bigotimes_{j=1}^{N}\left[\frac{d+2}{2}O_j^{\top}|b_j\rangle\langle b_j|O_j-\frac{1}{2}\eins_{d}\right]$, where $\vr$ is the PTI states, i.e. $\vr=\vr^{\top_j}$ for any $j$. See details in Appendix~\ref{E_Generalized}.

To test Observation~\ref{Ob:CS}, we consider the five-qubit noisy state $\vr_{p}=(1-p)|{\mathrm{GHZ}}_{+}\rangle\langle {\mathrm{GHZ}}_{+}|+p|{\mathrm{GHZ}}_{-}\rangle\langle{\mathrm{GHZ}}_{-}|$ with $p\in[0,1]$, where $|{\mathrm{GHZ}}_{\pm}\rangle=(|0\rangle^{\otimes 5}\pm|1\rangle^{\otimes 5})/\sqrt{2}$. We estimate the exact fidelity $f_{\mathrm{ex}} = \tr({\vr}_p|{\mathrm{GHZ}}_{+}\rangle\langle{\mathrm{GHZ}}_{+}|) = 1-p$ by calculating the quantity $f_{\mathrm{es}}=\tr(\tilde{\vr}_p|{\mathrm{GHZ}}_{+}\rangle\langle{\mathrm{GHZ}}_{+}|)$, where $\tilde{\vr}_p$ is the classical shadow of $\vr_{p}$. In Figs.~\ref{Fig3}(a) and~\ref{Fig3}(b), we plot the average of error $\mathcal{E}=|f_{\mathrm{es}}-f_{\mathrm{ex}}|$ for a five-qubit $\vr_p$ with classical shadows obtained from RMs and RRMs by using (a) global and (b) local orthogonal and unitary evolutions. The simulation results show that RRMs obtain a more accurate fidelity estimation than RMs with a fixed number of random operations.

\begin{figure}[t]
    \centering
    \includegraphics[scale=0.5]{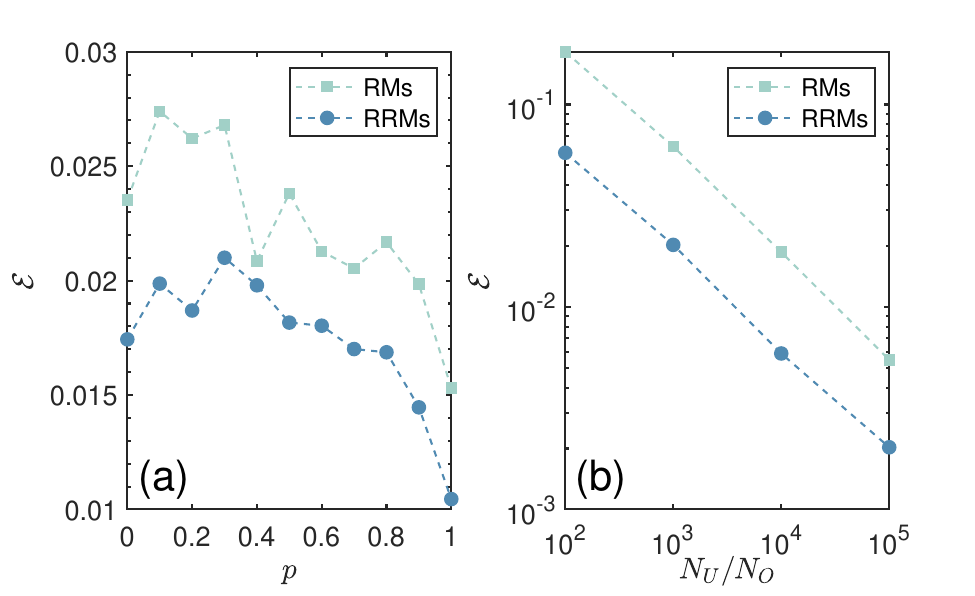}
    \caption{(a) The average of error $\mathcal{E}$ as a function of probability $p$ under $N_U=N_O=2000$ global unitary/orthogonal matrices by averaging over $100$ simulated experimental runs. (b) The average of error $\mathcal{E}$ as a function of the number of random local unitary or orthogonal matrices by averaging over $100$ simulated experimental runs and setting $p=0.5$.}
    \label{Fig3}
\end{figure}

\section{Experimental considerations}
Our protocols have advantages related to the experimental settings compared to standard RMs. One is that fewer experimental steps are required for performing an orthogonal matrix. In photonic systems, $1/2$ fewer wave plates are used compared with unitary operation~\cite{wu2021resource}. Another is that fewer random operations are enough for certain tasks. Fig.~\ref{Fig3}(b) shows that RRMs only require $10^4$ orthogonal matrices for attaining the estimation precision $0.005$, while $10^5$ unitaries are required for RMs.

An important step in implementing our protocols is to generate Haar random orthogonal matrices following the known algorithms~\cite{kothari2015almost,o2023explicit}. To check whether the sample matrices are indeed Haar random, one can apply the frame potential~\cite{renes2004symmetric,renes2004frames,gross2007evenly} for orthogonal matrices~\cite{hashagen2018real}.

\section{Conclusion}
We introduced the concepts of RRMs and PRRMs to characterize high-dimensional quantum states. We showed that partial correlations obtained by RRMs can detect bound entanglement and dimensionality, while the combination of RMs, RRMs, and PRRMs can characterize bipartite imaginary states. The latter allowed us to measure imaginarity without having shared reference frames. Furthermore, we extended the classical shadow by using random orthogonal matrices.
	
Several potential avenues for further investigation arise from our research. First, while we considered the orthogonal versions of RMs, developing RMs within different subsets of unitaries would be interesting. This could motivate further investigations of quantum resources beyond entanglement and imaginarity. Second, it would be worthwhile to find stronger entanglement criteria than the criterion by RMs itself, e.g., using linear or nonlinear combination in the space of $Q_{AB}^{(t)}$ and $\hat{Q}_{AB}^{(t)}$, or by taking into account an extension to non-product observables~\cite{wyderka2023complete}. Finally, it would be desirable to verify our results in different experimental platforms such as trapped ions~\cite{brydges2019probing}, Rydberg atoms~\cite{notarnicola2023randomized}, and photonic systems~\cite{wyderka2023complete,zhang2023experimental}.

\textit{Note added}. After finishing the work, we became aware that related results on real classical shadows were derived in Ref.~\cite{west2025real}.

\section*{ACKNOWLEDGMENTS}
We thank Dmitry Grinko, Jiajie Guo, Zhenhuan Liu, Feng-Xiao Sun, Nikolai Wyderka, Yu Xiang, and Xiao-Dong Yu for useful discussions. This work is supported by the National Natural Science Foundation of China (Grant No. 12125402, No. 12347152, No. 12405005, No. 12405006, No. 12075159, and No. 12171044), Beijing Natural Science Foundation (Grant No. Z240007), the Innovation Program for Quantum Science and Technology (No. 2024ZD0302401), the Postdoctoral Fellowship Program of CPSF (No. GZC20230103), the China Postdoctoral Science Foundation (No. 2023M740118 and 2023M740119), the specific research fund of the Innovation Platform for Academicians of Hainan Province, the Deutsche Forschungsgemeinschaft (DFG, German Research Foundation, Projects No. 447948357 and No. 440958198), the Sino-German Center for Research Promotion (Project M-0294), and the German Ministry of Education and Research (Project QuKuK, BMBF Grant No.\ 16KIS1618K). Satoya Imai acknowledges the European Commission through 397 the H2020 QuantERA ERA-NET Cofund in Quantum Technologies projects “MENTA”, Horizon Europe programme HORIZON-CL4-2022-QUANTUM-02-SGA via the project 101113690 (PASQuanS2.1), and JST ASPIRE (JPMJAP2339).

\section*{DATA AVAILABILITY}
The data that support the findings of this article are openly available: https://github.com/jmliang24/RRMs.git.

\appendix
\section{Preliminaries and calculation of second-order moments}\label{A_Preliminaries}
\subsection{Moments over orthogonal group}\label{Pre:moments}
Let us start by introducing the notation used in this work. We denote the set of linear operators that act on the $d$-dimensional complex vector space $\mathds{C}^{d}$ as $\mathcal{L}(\mathds{C}^{d})$. The identity operator with dimension $d$ is represented by $\eins_d$. The orthogonal group, denoted as $\boldsymbol{\mathrm{O}}(d)$, consists of operators $O\in\mathcal{L}(\mathds{C}^{d})$ that satisfy the conditions $OO^{\top}=O^{\top}O=\eins_d$. It is worth noting that the orthogonal group $\boldsymbol{\mathrm{O}}(d)$ is compact, which implies the existence of a unique probability measure $\mu$ that remains invariant under left and right group multiplication. This measure is commonly referred to as the \textit{Haar measure}, and we will denote integrals to the Haar measure as $\int dO$.

By introducing two important operators, the SWAP operator, $\mathbb{S}=\sum_{j,k=0}^{d-1}|jk\rangle\langle kj|$, and the operator $\Pi=\sum_{j,k=0}^{d-1}|jj\rangle\langle kk|$, we here present the calculation of the Haar integrals over the orthogonal group.
\begin{lemma}\label{Lemma1}
(First and second moment). Given $\mathcal{A}_1\in\mathcal{L}(\mathds{C}^{d})$ and $\mathcal{A}_2\in\mathcal{L}(\mathds{C}^{d}\otimes\mathds{C}^{d})$, we have
\begin{align}
    \int dO O\mathcal{A}_1O^{\top}&=\frac{\tr(\mathcal{A}_1)}{d}\eins_d,\label{FirstOrder}\\
    \int dO \,O^{\otimes2}\mathcal{A}_2O^{\top\otimes2}&=\gamma_1\eins_d\otimes\eins_d+\gamma_2\mathbb{S}+\gamma_3\Pi,\label{SecondOrder}
\end{align}
where the coefficients are
\begin{align}
    \gamma_1&=\frac{(d+1)\tr(\mathcal{A}_2)-\tr(\mathbb{S}\mathcal{A}_2)-\tr(\Pi\mathcal{A}_2)}{d(d-1)(d+2)},\\
    \gamma_2&=\frac{-\tr(\mathcal{A}_2)+(d+1)\tr(\mathbb{S}\mathcal{A}_2)-\tr(\Pi\mathcal{A}_2)}{d(d-1)(d+2)},\\
    \gamma_3&=\frac{-\tr(\mathcal{A}_2)-\tr(\mathbb{S}\mathcal{A}_2)+(d+1)\tr(\Pi\mathcal{A}_2)}{d(d-1)(d+2)}.
\end{align}
\end{lemma}
\begin{proof}
The first-order moment is proportional to the identity operator, i.e. $\int dOO\mathcal{A}_1O^{\top}=\alpha\eins_d$ with $\alpha\in\mathds{C}$. Taking the trace on both sides, we deduce that $\alpha=\frac{\tr(\mathcal{A}_1)}{d}$.

The second moment is a linear combination of three operators $\eins_d\otimes\eins_d$, $\mathbb{S}$, and $\Pi$ \cite{audenaert2002asymptotic},
\begin{align}
    \int\!dO(O\!\otimes\!O)\mathcal{A}_2(O^{\top}\!\otimes\!O^{\top})
    \!=\!\gamma_1\eins_d\!\otimes\!\eins_d\!+\!\gamma_2\mathbb{S}\!+\!\gamma_3\Pi.
\end{align}
To find the coefficients, we left-multiply both sides by $\eins_{d}\otimes\eins_d$, $\mathbb{S}$, $\Pi$ respectively, and take a trace, which gives us the following linear system of equations
\begin{align}
\begin{pmatrix}
    \tr(\mathcal{A}_2)\\
    \tr(\mathbb{S}\mathcal{A}_2)\\
    \tr(\Pi\mathcal{A}_2)
\end{pmatrix}=
    d\begin{pmatrix}
    d	&	1	&	1\\
    1	&	d	&	1\\
    1	&	1	&	d
    \end{pmatrix}
    \begin{pmatrix}
    \gamma_1\\
    \gamma_2\\
    \gamma_3
    \end{pmatrix}.
\end{align}
During the calculation, we use the fact that $\tr(\eins_d\otimes\eins_d)=d^2$ and $\tr(\mathbb{S})=\tr(\Pi)=d$. Solving this linear system, we obtain the coefficients and complete the proof.
\end{proof}
		
In general, the $t$-th moment operator can be expressed as a linear combination of $\frac{(2t)!}{2^tt!}$ operators, which are derived from the Brauer algebra \cite{brauer1937algebras} and the Schur-Weyl duality for orthogonal groups \cite{collins2006integration}. For instance, the fourth moment can be represented as a linear combination of $105$ operators.
		
\begin{lemma}\label{Lemma2}
    For any two operators $\mathcal{A},\mathcal{B}\in\mathcal{L}(\mathds{C}^{d})$, we have $\tr(\Pi\mathcal{A}\otimes\mathcal{B})=\tr(\mathbb{S}\mathcal{A}\otimes\mathcal{B}^{\top})$.
\end{lemma}
\begin{proof}
    Note that the operator $\Pi$ can be expressed as
    \begin{align}
        \Pi&=\sum_{j,k=0}^{d-1}|jj\rangle\langle kk|
        =\sum_{j,k=0}^{d-1}|j\rangle\langle k|\otimes|j\rangle\langle k|\nonumber\\
        &=\sum_{j,k=0}^{d-1}|j\rangle\langle k|\otimes(|k\rangle\langle j|)^{\top}=\mathbb{S}^{\top_2},
    \end{align}
where $A^{\top}$ denotes the transpose of operator $A$ and $\top_2$ is the partial transpose for the second system. Thus, we complete the proof via
    \begin{align}
        \tr(\Pi\mathcal{A}\otimes\mathcal{B})&=\tr(\mathbb{S}^{\top_2}\mathcal{A}\otimes\mathcal{B})
        =\tr(\mathbb{S}\mathcal{A}\otimes\mathcal{B}^{\top}).
    \end{align}
\end{proof}
		
\subsection{Useful properties of the generalized Gell-Mann matrices}\label{Pre:useful}
The generalized Gell-Mann matrices (GGMs) are extensions of the Gell-Mann matrices, which are originally defined for qutrit systems. The GGMs are denoted as $\{\lambda_j\}_{j=0}^{d^2-1}$, where $\lambda_0=\eins_d$ and $\lambda_j$ are the normalized Gell-Mann matrices. These matrices satisfy the following properties, $\lambda_{j}=\lambda_{j}^{\dag}$, $\tr(\lambda_{j}\lambda_{k})=d\delta_{jk}$, $\tr(\lambda_{j})=0$ for $j>0$ \cite{gell2018symmetries}. The GGMs consist of three different types of matrices \cite{bertlmann2008bloch},
		
(i) $\frac{d(d-1)}{2}$ symmetric GGMs,
\begin{align}
    \lambda_{jk}^{(s)}=\sqrt{\frac{d}{2}}\left(|j\rangle\langle k|+|k\rangle\langle j|\right),\!0\leq j<k\leq d-1.
\end{align}
		
(ii) $\frac{d(d-1)}{2}$ antisymmetric GGMs,
\begin{align}
    \lambda_{jk}^{(a)}\!=\!\sqrt{\frac{d}{2}}\left(-i|j\rangle\langle k|+i|k\rangle\langle j|\right),\!0\leq\!j<k\leq\!d-1.
\end{align}

(iii) $d-1$ diagonal GGMs,
\begin{align}
    \lambda_{l}^{(d)}=\mathcal{K}\Big[-(l+1)|l+1\rangle\langle l+1|+\sum_{j=0}^{l}|j\rangle\langle j|\Big],
\end{align}
where $\mathcal{K}=\sqrt{\frac{d}{(l+1)(l+2)}}$ and $0\leq l\leq d-2$.
	
It is important to note that all symmetric and diagonal Gell-Mann matrices have real elements and satisfy Assumption (ii) in the main text. On the other hand, all antisymmetric Gell-Mann matrices have imaginary elements and satisfy Assumption (iii) in the main text. As a result, we can divide the Gell-Mann matrices into two categories: real Gell-Mann matrices denoted as $\boldsymbol{\lambda}=\{\lambda_{1},\cdots,\lambda_{L}\}$ and imaginary Gell-Mann matrices denoted as $\hat{\boldsymbol{\lambda}}=\{\lambda_{L+1},\cdots,\lambda_{d^2-1}\}$. Here, $L=\frac{(d-1)(d+2)}{2}$ and $\hat{L}=\frac{d(d-1)}{2}$. It is worth mentioning that the total number of Gell-Mann matrices is given by $L+\hat{L}=d^2-1$. For the Gell-Mann matrices, we present the following lemma.
\begin{lemma}\label{Lemma3}
    For two GGMs $\mathcal{A}$ and $\mathcal{B}$ in Hilbert space $\mathcal{H}\in\mathbb{C}^{d\times d}$, we have
    \begin{align}
        \tr(\Pi\mathcal{A}\otimes\mathcal{B})=
        \begin{cases}
            d, & \mathcal{A}=\mathcal{B}\in\boldsymbol{\lambda}\\
            -d, & \mathcal{A}=\mathcal{B}\in\hat{\boldsymbol{\lambda}}\\
            0, & otherwise.\\
		\end{cases}.
    \end{align}
\end{lemma}
\begin{proof}
    Based on Lemma~\ref{Lemma2}, we immediately have
    \begin{align}
        \tr(\Pi\mathcal{A}\!\otimes\!\mathcal{B})
        \!=\!\tr(\mathbb{S}\mathcal{A}\!\otimes\!\mathcal{B}^{\top})
        \!=\!\begin{cases}
        \tr(\mathcal{A}\mathcal{B}),\!&\!\mathcal{B}\in\boldsymbol{\lambda}\\
        -\tr(\mathcal{A}\mathcal{B}),\!&\!
        \mathcal{B}\in\hat{\boldsymbol{\lambda}}\\
        \end{cases}.
    \end{align}
The proof is completed by using $\tr(\lambda_j\lambda_k)=d\delta_{jk}$ for two GGMs $\lambda_j$, $\lambda_k$.
\end{proof}
		
\subsection{Moments of RRMs and PRRMs}\label{Moments}
In this subsection, we derive the second moment for quantum states of single-qudit, two-qudit, and $N$-qudit systems.
		
\subsubsection{Single-qudit systems}\label{Moments:single}
Consider a single qudit state $\vr=\frac{1}{d}\sum_{j=0}^{d^2-1}T_j\lambda_j$. The second moment with real observable $M$, satisfying the conditions $\tr(M)=0$ and $\tr(M^2)=d$, is
\begin{align}
    Q^{(2)}&\!=\!\int\!dO\left[\tr\left(O\vr O^{\top}M\right)\right]^2
    \!=\!\frac{1}{d^2}\!\sum_{j,k=1}^{d^2-1}\!T_jT_k\tr(\Phi),
\end{align}
where the quantity
\begin{align}
    \tr(\Phi)&\!=\!\tr\left[\int dO\left(O\otimes O\right)(\lambda_j\otimes\lambda_k)\left(O^{\top}\otimes O^{\top}\right)\left(M\!\otimes\!M\right)\right]\nonumber\\
    &\!=\!\gamma_2\tr(\mathbb{S}M\otimes M)+\gamma_3\tr(\Pi M\otimes M)\nonumber\\
    &\!=\!d\left(\gamma_2+\gamma_3\right).
\end{align}
The third equation uses the fact that $\tr(\mathbb{S}M\otimes M)=\tr(\Pi M\otimes M)=\tr(M^2)=d$. The coefficients
\begin{align}
    \gamma_2&=\frac{(d+1)d\delta_{jk}}{d(d-1)(d+2)}-\frac{\tr(\Pi\lambda_j\otimes\lambda_k)}{d(d-1)(d+2)},\\
    \gamma_3&=-\frac{d\delta_{jk}}{d(d-1)(d+2)}+\frac{(d+1)\tr(\Pi\lambda_j\otimes\lambda_k)}{d(d-1)(d+2)}.
\end{align}
Following Lemma~\ref{Lemma1}, the quantity
\begin{align}
    \tr(\Phi) &=d\left(\gamma_2+\gamma_3\right)
    =\frac{d^2\delta_{jk}+d\tr(\Pi\lambda_j\otimes\lambda_k)}{(d-1)(d+2)}\nonumber\\
    &=\begin{cases}
        \frac{2d^2\delta_{jk}}{(d-1)(d+2)}, & \lambda_k\in\boldsymbol{\lambda}\\
        0, & \lambda_k\in\hat{\boldsymbol{\lambda}}\\
    \end{cases}.
\end{align}
Using Lemma~\ref{Lemma3}, we finally obtain the second moment of RRMs, 
    \begin{align}
        Q^{(2)}&=\frac{1}{d^2}\times\frac{2d^2}{(d-1)(d+2)}\sum_{j=1}^{L}T_j^2\nonumber\\
        &=\frac{2}{(d-1)(d+2)}\sum_{j=1}^{L}T_j^2=L^{-1}\sum_{j=1}^{L}T_j^2.
    \end{align}
		
If we utilize an imaginary observable $\hat{M}$, such that $\tr(\hat{M})=0$ and $\tr(\hat{M}^2)=d$, the quantity
\begin{align}
    \tr(\Phi) &=d\left(\gamma_2-\gamma_3\right)
    =\frac{d\delta_{jk}-\tr(\Pi\lambda_j\otimes\lambda_k)}{d-1}\nonumber\\
    &=\begin{cases}
        0, & \lambda_k\in\boldsymbol{\lambda}\\
        \frac{2d}{d-1}, & \lambda_k\in\hat{\boldsymbol{\lambda}}\\
    \end{cases}.
\end{align}
Using the same tricks as RRMs, the second moment of PRRMs is
\begin{align}
    \hat{Q}^{(2)}&=\int dO\left[\tr\left(O\vr O^{\top}\hat{M}\right)\right]^2 \nonumber\\
    &=\frac{1}{d^2}\times\frac{2d}{d-1}\sum_{j=L+1}^{L+\hat{L}}T_j^2\nonumber\\
    &=\frac{2}{d(d-1)}\sum_{j=L+1}^{L+\hat{L}}T_j^2\nonumber\\
    &=\hat{L}^{-1}\sum_{j=L+1}^{L+\hat{L}}T_j^2.
\end{align}
		
\subsubsection{Two-qudit system}\label{Moments:two}
Consider a two-qudit state $\vr=\frac{1}{d^2}\sum_{j,k=0}^{d^2-1}T_{jk}\lambda_j^A\otimes\lambda_k^B$. The second moment with real observables $M_A$ and $M_B$ satisfying $\tr(M_s)=0$ and $\tr(M_s^2)=d$ for $s=A,B$ is
\begin{align}
    &Q^{(2)}\nonumber\\
    &=\!\int\!dO_A\!\int\!dO_B\left\{\tr\left[(O_A\!\otimes\!O_B)\vr (O_A^{\top}\!\otimes\!O_B^{\top})(M_A\!\otimes\! M_B)\right]\right\}^{2}\nonumber\\
    &=\frac{1}{d^4}\sum_{j_{1},k_{1},j_{2},k_{2}=1}^{d^2-1}T_{j_{1}k_{1}}T_{j_{2}k_{2}}\tr(\Phi_A)\tr(\Phi_B).
\end{align}
where the quantity
\begin{align}
    \tr(\Phi_A)&=\tr\left[\int dO_AO_A^{\otimes2}(\lambda_{j_{1}}^A\otimes\lambda_{j_2}^A)O_A^{\top\otimes2}M_A^{\otimes2}\right]\nonumber\\
    &=\gamma_2\tr\left[\mathbb{S}(M_{A}\otimes M_{A})\right]+\gamma_3\tr\left[\Pi(M_{A}\otimes M_{A})\right]\nonumber\\
    &=d(\gamma_2+\gamma_3).
\end{align}
A similar result for $\tr(\Phi_B)$ can also be found. Following the derivation of the single-qubit case, the second moment is
\begin{align}
    Q^{(2)}&=\frac{1}{d^4}\left[\frac{2d^2}{(d-1)(d+2)}\right]^2\sum_{j,k=1}^{L}T_{jk}^2\nonumber\\
    &=\left[\frac{2}{(d-1)(d+2)}\right]^2\sum_{j,k=1}^{L}T_{jk}^2\nonumber\\
    &=L^{-2}\sum_{j,k=1}^{L}T_{jk}^2.
\end{align}

If we utilize imaginary observables $\hat{M}_A$, $\hat{M}_B$ satisfying $\tr(\hat{M}_s)=0$ and $\tr(\hat{M}_s^2)=d$ for $s=A,B$, the second moment of PRRMs turns to
\begin{align}
    &\hat{Q}^{(2)}\nonumber\\
    &\!=\!\int dO_A\!\int dO_B\left\{\tr\left[(O_A\!\otimes\!O_B)\vr (O_A^{\top}\!\otimes\!O_B^{\top})(\hat{M}_A\!\otimes\!\hat{M}_B)\right]\right\}^{2}\nonumber\\
    &\!=\!\frac{1}{d^4}\left[\frac{2d}{d-1}\right]^2\!\sum_{j,k=L+1}^{L+\hat{L}}T_{jk}^2\nonumber\\
    &\!=\!\hat{L}^{-2}\!\sum_{j,k=L+1}^{L+\hat{L}}T_{jk}^2.
\end{align}
		
\subsubsection{Multipartite systems}\label{Moments:multi}
Generalizing the above results to an $N$-qudit state $\vr=\frac{1}{d^N}\sum_{j_1,\cdots,j_N=0}^{d^2-1}T_{j_1\cdots j_N}\lambda_{j_1}\otimes\cdots\otimes\lambda_{j_N}$. The second moment with real observables $M_s$ satisfying $\tr(M_s)=0$ and $\tr(M_s^2)=d$ for $s=1,\cdots,N$ is
\begin{align}
    Q^{(2)}&=\frac{1}{d^{2N}}\left[\frac{2d^2}{(d-1)(d+2)}\right]^N\sum_{j_1,\cdots,j_N=1}^{L}T_{j_1\cdots j_N}^2\nonumber\\
    &=L^{-N}\sum_{j_1,\cdots,j_N=1}^{L}T_{j_1\cdots j_N}^2.
\end{align}
		
If we use imaginary observables $\hat{M}_s$ satisfying $\tr(\hat{M}_s)=0$ and $\tr(\hat{M}_s^2)=d$ for $s=1,\cdots,N$, the second moment of PRRMs turns to
\begin{align}
    \hat{Q}^{(2)}&=\frac{1}{d^{2N}}\left[\frac{2d}{d-1}\right]^N\sum_{j_1,\cdots,j_N=L+1}^{L+\hat{L}}T_{j_1\cdots j_N}^2\nonumber\\
    &=\left[\frac{2}{d(d-1)}\right]^N\sum_{j_1,\cdots,j_N=L+1}^{L+\hat{L}}T_{j_1\cdots j_N}^2\nonumber\\
    &=\hat{L}^{-N}\sum_{j_1,\cdots,j_N=L+1}^{L+\hat{L}}T_{j_1\cdots j_N}^2.
\end{align}
		
\section{Proofs of the observations in the main text}\label{B_Proofs}
In this section, we present the observations in the main text and give the corresponding proofs.

\subsection{Proof for observation 1}
For $t=2$, the proof is straightforward by combining the second moments of a bipartite state described in Appendix~\ref{Moments:two}.
			
For general $t$, the proof follows the definition of the $t$-order orthogonal moment $S^{(t)}$ over a generalized pseudo-Bloch sphere in Appendix~\ref{C_Estimating}. A similar computational approach can also be found in the proof of Theorem 1 in Ref.~\cite{wyderka2023probing}. For real observables, the rotations always occur on the real GMs and thereby only extract real-real correlations. Imaginary observables have similar results.
\subsection{Proof for observation 2}
For a fixed $Q_{AB}^{(2)}=C_2$, we find the minimum of the optimization problem, $f_{\min}(r,C_2)$ for $r=1,2,\cdots,d$. For all separable states, we set $r=1$ and find the corresponding minimum $f_{\min}(1,C_2)$. If the measured fourth-moment of a state $\vr$, $C_4$, satisfies, $C_4<f_{\min}(1,C_2)$, then $\vr$ is an entangled state. Using the same trick, we complete the proof by considering other $r=2,\cdots,d$.

\subsection{Proof for observation 3}
We express $\vr_{AB}$ as the following decomposition
\begin{align}
    \vr_{AB}&=\frac{1}{d^2}\sum_{j,k=0}^{d^2-1}T_{jk}\lambda_j^A\otimes\lambda_k^B\nonumber\\
    &=\frac{1}{d^2}\eins_d^{A}\otimes\eins_d^{B}+\frac{1}{d^2}\sum_{k=1}^{d^2-1}T_{0k}\eins_d^{A}\otimes\lambda_k^B\nonumber\\
    &+\frac{1}{d^2}\sum_{j=1}^{d^2-1}T_{j0}\lambda_j^A\otimes\eins_d^{B}+\frac{1}{d^2}\sum_{j,k=1}^{d^2-1}T_{jk}\lambda_j^A\otimes\lambda_k^B\nonumber\\
    &=\vr_A\otimes\frac{\eins_d^B}{d}+\frac{\eins_d^A}{d}\otimes\vr_B+\frac{1}{d^2}\sum_{j,k=1}^{d^2-1}T_{jk}\lambda_j^A\otimes\lambda_k^B\nonumber\\
    &-\frac{1}{d^2}\eins_d^{A}\otimes\eins_d^{B},
\end{align}
where the reduced states
\begin{align}
    \vr_A&=\tr_B(\vr_{AB})=\frac{1}{d}\Big(\eins_d^A+\sum_{j=1}^{d^2-1}T_{j0}\lambda_j^A\Big),\nonumber\\
    \vr_B&=\tr_A(\vr_{AB})=\frac{1}{d}\Big(\eins_d^B+\sum_{k=1}^{d^2-1}T_{0k}\lambda_k^B\Big).
\end{align}
The imaginarity of $\vr_{AB}$ is characterized by the imaginarity of reduced states $\vr_A$, $\vr_B$, and the correlation $\sum_{j,k=1}^{d^2-1}T_{jk}\lambda_j^A\otimes\lambda_k^B$. Thus, $\vr_{AB}$ is imaginary if and only if one of the following conditions is true,
			
(a) $\vr_A$ is imaginary,
			
(b) $\vr_B$ is imaginary,
			
(c) the correlation $\sum_{j,k=1}^{d^2-1}T_{jk}\lambda_j^A\otimes\lambda_k^B$ contains imaginarity. One of $\lambda_j^A$ and $\lambda_k^B$ in the correlation $\sum_{j,k=1}^{d^2-1}T_{jk}\lambda_j^A\otimes\lambda_k^B$ is real GGMs and the other is imaginary GGMs.
			
The above conditions correspond to the quantities (a) $\hat{Q}_A^{(2)}>0$, (b) $\hat{Q}_B^{(2)}>0$, and (c) $G_{AB}^{(2)}>0$, respectively. Conversely, $\vr_{AB}$ is a real state if and only if all conditions are not true. Thus, the result (i) is proved.
			
Next, we prove the lower bound of the robustness of imaginarity of $\vr_{AB}$. Based on Ref.~\cite{wu2021operational}, we have
\begin{align}
    \mathcal{F}_{R}(\vr_{AB})&=\min_{\tau}\left\{\mu\geq0:\frac{\vr_{AB}+\mu\tau}{1+\mu}\in\mathscr{R}\right\}\\
    &=\frac{1}{2}\|\vr_{AB}-\vr_{AB}^{\top}\|_{\tr}\\
    &\geq\frac{1}{2}\|\vr_{AB}-\vr_{AB}^{\top}\|_{\mathrm{HS}},
\end{align}
where $\|A\|_{\tr}=\tr\left(\sqrt{A^{\dag}A}\right)=\sum_{j}\sigma_j$ denotes the trace norm (the sum of singular values) of $A$ and $\|A\|_{\mathrm{HS}}=\sqrt{\tr\left(A^{\dag}A\right)}=\sqrt{\sum_j\sigma_j^2}$ is the Hilbert-Schmidt norm (the square root of the sum of the squares of its singular values) of $A$. Note that the inequality $\|A\|_{\tr}\geq\|A\|_{\mathrm{HS}}$ holds, and the equality is true if and only if the operator $A$ has all zero singular values.
			
Motivated by the Bloch decomposition of $\vr_{AB}$, we express the correlations in $\vr_{AB}$ as four parts, including real-real, imaginary-imaginary, real-imaginary, and imaginary-real parts. Here, we assume the first $L$ Gell-Mann matrices are real GMMs and next $\hat{L}$ are imaginary GMMs. We then obtain the following equations,
\begin{align}
    \frac{1}{d^2}\sum_{j,k=1}^{d^2-1}T_{jk}\lambda_j^A\otimes\lambda_k^B
    &=\frac{1}{d^2}\sum_{j,k=1}^{L}T_{jk}\lambda_j^A\otimes\lambda_k^B\nonumber\\
    &+\frac{1}{d^2}\sum_{j,k=L+1}^{L+\hat{L}}T_{jk}\lambda_j^A\otimes\lambda_k^B\nonumber\\
    &+\frac{1}{d^2}\sum_{j=1}^{L}\sum_{k=L+1}^{L+\hat{L}}T_{jk}\lambda_j^A\otimes\lambda_k^B\nonumber\\
    &+\frac{1}{d^2}\sum_{j=L+1}^{L+\hat{L}}\sum_{k=1}^{L}T_{jk}\lambda_j^A\otimes\lambda_k^B,
\end{align}
and its transposition
\begin{align}
    \Big(\frac{1}{d^2}\!\sum_{j,k=1}^{d^2-1}T_{jk}\lambda_j^A\!\otimes\!\lambda_k^B\Big)\!^{\top}
    &\!=\!\frac{1}{d^2}\sum_{j,k=1}^{L}T_{jk}\lambda_j^A\!\otimes\!\lambda_k^B\nonumber\\
    &\!+\!\frac{1}{d^2}\sum_{j,k=L+1}^{L+\hat{L}}T_{jk}\lambda_j^A\!\otimes\!\lambda_k^B\nonumber\\
    &\!-\!\frac{1}{d^2}\sum_{j=1}^{L}\sum_{k=L+1}^{L+\hat{L}}T_{jk}\lambda_j^A\!\otimes\!\lambda_k^B\nonumber\\
    &\!-\!\frac{1}{d^2}\sum_{j=L+1}^{L+\hat{L}}\sum_{k=1}^{L}T_{jk}\lambda_j^A\!\otimes\!\lambda_k^B.
\end{align}
The gap between the above quantities is given by
\begin{align}
    &\frac{1}{d^2}\sum_{j,k=1}^{d^2-1}T_{jk}\lambda_j^A\!\otimes\!\lambda_k^B-\Big(\frac{1}{d^2}\sum_{j,k=1}^{d^2-1}T_{jk}\lambda_j^A\!\otimes\!\lambda_k^B\Big)^{\top}\nonumber\\
    &=\sum_{j=1}^{L}\!\sum_{k=L+1}^{L+\hat{L}}\frac{2T_{jk}}{d^2}\lambda_j^A\!\otimes\!\lambda_k^B\!+\!\sum_{j=L+1}^{L+\hat{L}}\!\sum_{k=1}^{L}\frac{2T_{jk}}{d^2}\lambda_j^A\!\otimes\!\lambda_k^B.
\end{align}
We then obtain
\begin{align}
    \vr_{AB}\!-\!\vr_{AB}^{\top}&\!=\!\left(\vr_A-\vr_A^{\top}\right)\otimes\frac{\eins_d^B}{d}+\frac{\eins_d^A}{d}\otimes\left(\vr_B-\vr_B^{\top}\right)\nonumber\\
    &+\sum_{j=1}^{L}\sum_{k=L+1}^{L+\hat{L}}\frac{2T_{jk}}{d^2}\lambda_j^A\otimes\lambda_k^B\nonumber\\
    &+\sum_{j=L+1}^{L+\hat{L}}\sum_{k=1}^{L}\frac{2T_{jk}}{d^2}\lambda_j^A\otimes\lambda_k^B.
\end{align}
Thus, the Hilbert-Schmidt distance is given by
\begin{widetext}
    \begin{align}
    \|\vr_{AB}\!-\!\vr_{AB}^{\top}\|_{\mathrm{HS}}^2
    &\!=\!\|\left(\vr_A\!-\!\vr_A^{\top}\right)\!\otimes\!\frac{\eins_d^B}{d}
    \!+\!\frac{\eins_d^A}{d}\!\otimes\!\left(\vr_B-\vr_B^{\top}\right)
    \!+\!\frac{2}{d^2}(\sum_{j=1}^{L}\sum_{k=L+1}^{L+\hat{L}}T_{jk}\lambda_j^{A}\!\otimes\!\lambda_k^{B}
    \!+\!\sum_{j=L+1}^{L+\hat{L}}\sum_{k=1}^{L}T_{jk}\lambda_j^{A}\!\otimes\!\lambda_k^{B})\|_{\mathrm{HS}}^2\nonumber\\
    &\!=\!\left\|\left(\vr_A\!-\!\vr_A^{\top}\right)\!\otimes\!\frac{\eins_d^B}{d}\right\|_{\mathrm{HS}}^2\!+\!\left\|\frac{\eins_d^A}{d}\!\otimes\!\left(\vr_B\!-\!\vr_B^{\top}\right)\right\|_{\mathrm{HS}}^2\!+\!\|\frac{2}{d^2}(\sum_{j=1}^{L}\sum_{k=L+1}^{L+\hat{L}}T_{jk}\lambda_j^{A}\!\otimes\!\lambda_k^{B}
    \!+\!\sum_{j=L+1}^{L+\hat{L}}\sum_{k=1}^{L}T_{jk}\lambda_j^{A}\!\otimes\!\lambda_k^{B})\|_{\mathrm{HS}}^2\nonumber\\
    &\!=\!\frac{1}{d}(\left\|\vr_B\!-\!\vr_B^{\top}\right\|_{\mathrm{HS}}^2
    \!+\!\left\|\vr_A\!-\!\vr_A^{\top}\right\|_{\mathrm{HS}}^2)
    \!+\!\frac{4}{d^4}
    \|\sum_{j=1}^{L}\sum_{k=L+1}^{L+\hat{L}}T_{jk}\lambda_j^{A}\!\otimes\!\lambda_k^{B}
    \!+\!\sum_{j=L+1}^{L+\hat{L}}\sum_{k=1}^{L}T_{jk}\lambda_j^{A}\!\otimes\!\lambda_j^{B}\|_{\mathrm{HS}}^2.
\end{align}
\end{widetext}
where the last two equations are derived by the definition of the Hilbert-Schmidt norm.

Next, we divide $\vr_A$ into three parts in terms of real GGMs and imaginary GGMs
\begin{align}
    \vr_A=\frac{1}{d}\eins_d^A+\frac{1}{d}\sum_{j=1}^{L}T_{j0}\lambda_{j}^A
    +\frac{1}{d}\sum_{k=L+1}^{L+\hat{L}}T_{k0}\lambda_k^A
\end{align}
and find the matrix
\begin{align}
    \vr_A-\vr_A^{\top}=\frac{2}{d}\sum_{k=L+1}^{L+\hat{L}}T_{k0}\lambda_k^A,
\end{align}
The square of the Hilbert-Schmidt norm of matrix $\vr_A-\vr_A^{\top}$ is
\begin{align}
    \|\vr_A-\vr_A^{\top}\|_{\mathrm{HS}}^2
    &=\|\frac{2}{d}\sum_{k=L+1}^{L+\hat{L}}T_{k0}\lambda_k^A\|_{\mathrm{HS}}^2\nonumber\\
    &=\tr[\frac{4}{d^2}\sum_{j,k=L+1}^{L+\hat{L}}T_{j0}T_{k0}\lambda_j^A\lambda_k^A]\nonumber\\
    &=\frac{4}{d}\sum_{j=L+1}^{L+\hat{L}}T_{j0}^2.
\end{align}
The gap between the second moments of RMs and RRMs is defined as
\begin{align}
    G_{A}&=(d^2-1)R_{A}^{(2)}-LQ_{A}^{(2)}=\sum_{j=1}^{d^2-1}T_{j0}^2-\sum_{k=1}^{L}T_{k0}^2\nonumber\\
    &=\sum_{j=L+1}^{L+\hat{L}}T_{j0}^2=\hat{L}\hat{Q}_A^{(2)}.
\end{align}
The gap $G_A$ coincides with the second moment of PRRMs $\hat{Q}_A^{(2)}$. Therefore, we have $\|\vr_A-\vr_A^{\top}\|_{\mathrm{HS}}^2=\frac{4\hat{L}}{d}\hat{Q}_A^{(2)}$ and similarly $\|\vr_B-\vr_B^{\top}\|_{\mathrm{HS}}^2=\frac{4\hat{L}}{d}\hat{Q}_B^{(2)}$.
			
Note that the two-qudit gap $G_{AB}^{(2)}$ is associated with the real-imaginary and imaginary-real correlations such that
\begin{widetext}
    \begin{align}
    \left\|\frac{2}{d^2}\left(\sum_{j=1}^{L}\sum_{k=L+1}^{L+\hat{L}}T_{jk}\lambda_j^{A}\otimes\lambda_k^{B}
    +\sum_{j=L+1}^{L+\hat{L}}\sum_{k=1}^{L}T_{jk}\lambda_j^{A}\otimes\lambda_k^{B}\right)\right\|_{\mathrm{HS}}^2
    &=\frac{4}{d^2}\left(\sum_{j=1}^{L}\sum_{k=L+1}^{L+\hat{L}}T_{jk}^2+\sum_{j=L+1}^{L+\hat{L}}\sum_{k=1}^{L}T_{jk}^2\right)\nonumber\\
    &=\frac{4}{d^2}\left[(d^2-1)^2R_{AB}^{(2)}-L^2Q_{AB}^{(2)}-\hat{L}^2\hat{Q}_{AB}^{(2)}\right]\nonumber\\
    &=\frac{4}{d^2}G_{AB}^{(2)}.
\end{align}
\end{widetext}
As a result, we obtain
\begin{align}
    \|\vr_{AB}-\vr_{AB}^{\top}\|_{\mathrm{HS}}^2
    \!=\!\frac{4}{d^2}\left[\hat{L}\left(\hat{Q}_A^{(2)}+\hat{Q}_B^{(2)}\right)+G_{AB}^{(2)}\right].
\end{align}
Finally, the robustness of the imaginarity measure $\mathcal{F}_{R}(\vr_{AB})$ has a lower bound
\begin{align}
    \mathcal{F}_{R}(\vr_{AB})&\geq\frac{1}{2}\|\vr_{AB}-\vr_{AB}^{\top}\|_{\mathrm{HS}}\nonumber\\
    &=\frac{1}{d}\sqrt{\hat{L}\left(\hat{Q}_A^{(2)}+\hat{Q}_B^{(2)}\right)+G_{AB}^{(2)}}.
\end{align}
The equality holds if and only if $\vr_{AB}$ is a real state.

\subsection{Proof for observation 4}
For the measurement channel $\mathcal{M}(\vr)$, we have
\begin{align}
    \mathcal{M}(\vr)
    &\!=\!\sum_{b=1}^{d^N}\int dO\langle b|O\vr O^{\top}|b\rangle O^{\top}|b\rangle\langle b|O\nonumber\\
    &\!=\!\sum_{b=1}^{d^N}\int dO\tr(\vr O^{\top}|b\rangle\langle b|O)O^{\top}| b\rangle\langle b|O\nonumber\\
    &\!=\!\sum_{b=1}^{d^N}\int dO\tr_1\left[(\vr\otimes\eins_{d^N})O^{\top\otimes2}| b\rangle\langle b|^{\otimes2}O^{\otimes2}\right]\nonumber\\
    &\!=\!\sum_{b=1}^{d^N}\!\tr_1[(\vr\otimes\eins_{d^N})\!\int\! dOO^{\top\otimes2}| b\rangle\langle b|^{\otimes2}O^{\otimes2}].
\end{align}
Utilizing Lemma~\ref{Lemma1}, we have
\begin{align}
    \mathcal{M}(\vr)&=\sum_{b=1}^{d^N}\frac{\tr_1\left[\vr\otimes\eins_{d^N}+(\vr\otimes\eins_{d^N})\mathbb{S}+(\vr\otimes\eins_{d^N})\Pi\right]}{d^N(d^N+2)}\nonumber\\
    &=\frac{\tr(\vr)\eins_{d^N}+\vr+\vr^{\top}}{d^N+2},\nonumber\\
    &=\frac{\eins_{d^N}+2\vr}{d^N+2}.
\end{align}
The second equality utilizes the following tricks (see Eqs. (140 and 152) in Ref.~\cite{mele2024introduction}),
\begin{align}
    &\tr_1(\mathcal{A}\otimes\mathcal{B}\mathbb{S})=\mathcal{B}\mathcal{A},
    \quad\tr_1(\mathcal{A}\otimes\mathcal{B}\Pi)=\mathcal{B}\mathcal{A}^{\top}.
\end{align}
Finally, we have that $\mathcal{M}^{-1}(\vr)=\frac{d^N+2}{2}\vr-\frac{1}{2}\eins_{d^N}$, since it can be easily checked that $\mathcal{M}^{-1}\left[\mathcal{M}(\vr)\right]=\vr$.
		
\section{Estimating the moments of bipartite systems from random rotations}\label{C_Estimating}
To evaluate the $t$-order moments of RMs and PRRMs, we introduce the $t$-order orthogonal moment $S^{(t)}$, which is taken by an integral over a generalized pseudo-Bloch sphere,
\begin{align}
\footnotesize
    S^{(t)}&\!=\!\frac{1}{\mathcal{V}^2}\!\iint\! d\boldsymbol{u}^Ad\boldsymbol{u}^B\left\{\tr\left[\vr \left(\boldsymbol{u}^A\boldsymbol{\lambda}^A\otimes\boldsymbol{u}^B\boldsymbol{\lambda}^B\right)\right]\right\}^{t},\nonumber\\
    &\!=\!\frac{1}{\mathcal{V}^2}\!\iint\!d\boldsymbol{u}^Ad\boldsymbol{u}^B[\sum_{j,k=1}^{L}\!\frac{T_{jk}}{d^2}\tr(\lambda_j^{A}\boldsymbol{u}^A\boldsymbol{\lambda}^A)\tr(\lambda_k^{B}\boldsymbol{u}^B\boldsymbol{\lambda}^B)]^{t}\nonumber\\
    &\!=\!\frac{1}{\mathcal{V}^2}\iint d\boldsymbol{u}^Ad\boldsymbol{u}^B\left[\sum_{l_1,l_2=1}^{L}T_{l_1l_2}u^A_{l_1}u^B_{l_2}\right]^{t},
\end{align}
where $\boldsymbol{u}^A=\left(u_1^A,\cdots,u_L^A\right)$ and $\boldsymbol{u}^B=\left(u_1^B,\cdots,u_L^B\right)$ are $L$-dimensional unit vectors. Here, the integral spans all $L$-dimensional unit vectors and $\mathcal{V}=\frac{2\pi^{\frac{L}{2}}}{\Gamma(\frac{L}{2})}$ is the surface of the unit sphere in $L$ dimensions~\cite{wyderka2023probing}, where $\Gamma(z)$ denotes Euler's gamma function with $\Gamma(z)=(z-1)!$ for positive integer $z$~\cite{folland2001integrate}. Using the multinomial theorem
\begin{align}
    (\sum_{j,k=1}^{n}X_{j,k})^t\!=\!\sum_{t_{1,1}+\cdots+t_{n,n}=t}\frac{t!}{t_{1,1}!\cdots t_{n,n}!}\prod_{j,k=1}^{n}X_{j,k}^{t_{j,k}},
\end{align}
we have
\begin{widetext}
    \begin{align}
    S^{(t)}=&\frac{1}{\mathcal{V}^2}\iint d\boldsymbol{u}^Ad\boldsymbol{u}^B\sum_{t_{1,1}+\cdots+t_{L,L}=t}\frac{t!}{t_{1,1}!\cdots t_{L,L}!}\prod_{l_1,l_2=1}^{L}\left(T_{l_1l_2}u^A_{l_1}u^B_{l_2}\right)^{t_{l_1,l_2}}\\
    =&\frac{1}{\mathcal{V}^2}\sum_{t_{1,1}+\cdots+t_{L,L}=t}\frac{t!}{t_{1,1}!\cdots t_{L,L}!}\prod_{l_1,l_2=1}^{L}T_{l_1l_2}^{t_{l_1,l_2}}\int d\boldsymbol{u}^A\prod_{l_1=1}^{L}\left(u^A_{l_1}\right)^{\alpha_{l_1}}\int d\boldsymbol{u}^B\prod_{l_2=1}^{L}\left(u^B_{l_2}\right)^{\tilde{\alpha}_{l_2}}.
\end{align}
\end{widetext}
For all odd $\alpha_{l_1}$ and $\tilde{\alpha}_{l_2}$, the integrals are $\int d\boldsymbol{u}^A\prod_{l_1=1}^{L}\left(u^A_{l_1}\right)^{\alpha_{l_1}}=\int d\boldsymbol{u}^B\prod_{l_2=1}^{L}\left(u^B_{l_2}\right)^{\tilde{\alpha}_{l_2}}=0$ \cite{folland2001integrate}. Thus, we only consider even numbers and yield
\begin{widetext}
    \begin{align}
    S^{(t)}=\frac{4}{\mathcal{V}^2}\sum_{t_{1,1}+\cdots+t_{L,L}=t}\frac{t!}{t_{1,1}!\cdots t_{L,L}!}\prod_{l_1,l_2=1}^{L}T_{l_1l_2}^{t_{l_1,l_2}}B(\beta_{1},\cdots,\beta_{L})B(\tilde{\beta}_{1},\cdots,\tilde{\beta}_{L}).
\end{align}
\end{widetext}
The multi-variable beta function is defined as,
\begin{align}
    B(\beta_{1},\cdots,\beta_{L})
    &=\frac{\Gamma(\beta_{1})\cdots\Gamma(\beta_{L})}{\Gamma(\beta_{1}+\cdots+\beta_{L})},\nonumber\\
    B(\tilde{\beta}_{1},\cdots,\tilde{\beta}_{L})&=\frac{\Gamma(\tilde{\beta}_{1})\cdots\Gamma(\tilde{\beta}_{L})}{\Gamma(\tilde{\beta}_{1}+\cdots+\tilde{\beta}_{L})},
\end{align}
and parameters are
\begin{align}
    \alpha_{l_1}&=\sum_{l_2=1}^{L}t_{l_1,l_2},~
    \tilde{\alpha}_{l_2}=\sum_{l_1=1}^{L}t_{l_1,l_2},\nonumber\\
    \beta_{l_1}&=\frac{\alpha_{l_1}+1}{2},~
    \tilde{\beta}_{l_2}=\frac{\tilde{\alpha}_{l_2}+1}{2},
\end{align}
where the gamma function $\Gamma(k)=(k-1)!$ and $\Gamma(k+1/2)=(k-1/2)\cdots\frac{1}{2}\sqrt{\pi}$ for positive number $k$~\cite{folland2001integrate}.
		
For $t=2$, $t_{1,1}+\cdots+t_{L,L}=2$ indicates that (i) only one item is $2$ and other items are zero; (ii) only two items are equal to one and other items are zero. In condition (ii), the parameters $\alpha_{l_1}$ and $\alpha_{l_2}$ are always zero or one respectively. Thus, the integral vanishes in this case. In condition (i), we find $\frac{2!}{0!\cdots 2!\cdots0!}=1$ and
\begin{align}
    B(\beta_{1},\cdots,\beta_{L})
    &\!=\!B(\tilde{\beta}_{1},\cdots,\tilde{\beta}_{L})
    =B\left(\frac{1}{2},\cdots,\frac{3}{2},\cdots,\frac{1}{2}\right)\nonumber\\
    &\!=\!\frac{\left[\Gamma\left(\frac{1}{2}\right)\right]^{L-1}\Gamma\left(\frac{3}{2}\right)}{\Gamma\left(\frac{L+2}{2}\right)}=\frac{\pi^{\frac{L}{2}}}{2\Gamma\left(\frac{L+2}{2}\right)}.
\end{align}
Note that $\Gamma(\frac{1}{2})=\sqrt{\pi}$. As a result, the second moment is
\begin{align}
    S^{(2)}
    &=\frac{4}{\mathcal{V}^2}B^2(\beta_{1},\cdots,\beta_{L})\sum_{l_1,l_2=1}^{L}T_{l_1l_2}^2\nonumber\\
    &=\frac{4}{\mathcal{V}^2}\left[\frac{\pi^{\frac{L}{2}}}{2\Gamma\left(\frac{L+2}{2}\right)}\right]^2\sum_{l_1,l_2=1}^{L}T_{l_1l_2}^2\nonumber\\
    &=\frac{1}{\mathcal{V}^2}\frac{\pi^{L}}{\Gamma^2\left(\frac{L+2}{2}\right)}\sum_{l_1,l_2=1}^{L}T_{l_1l_2}^2
    =V_1\sum_{l_1,l_2=1}^{L}T_{l_1l_2}^2,
\end{align}
where the constant
\begin{align}
    V_1\!=\!\frac{1}{\mathcal{V}^2}\frac{\pi^{L}}{\Gamma^2\left(\frac{L+2}{2}\right)}
    \!=\!\frac{\Gamma^2(\frac{L}{2})\pi^{L}}{4\pi^L\Gamma^2\left(\frac{L+2}{2}\right)}
    \!=\!\frac{\Gamma^2(\frac{L}{2})}{4\Gamma^2(\frac{L+2}{2})}.
\end{align}
		
For $t=4$, we consider the following cases such that $t_{1,1+\cdots+t_{L,L}}=4$.
		
(i) Only one $t_{l_1,l_2}=4$ and other items are zero. In this case, we have
\begin{align}
    \alpha_{l_1}&\!=\!4,~
    \tilde{\alpha}_{l_2}=4,~
    \beta_{l_1}=\frac{5}{2},~
    \tilde{\beta}_{l_2}=\frac{5}{2},\nonumber\\
    B(\beta_{1},\cdots,\beta_{L})&=B(\tilde{\beta}_{1},\cdots,\tilde{\beta}_{L})
    =B\left(\frac{1}{2},\cdots,\frac{5}{2},\cdots,\frac{1}{2}\right)\nonumber\\
    &\!=\!\frac{\left[\Gamma\left(\frac{1}{2}\right))\right]^{L-1}\Gamma\left(\frac{5}{2}\right)}{\Gamma\left(\frac{L+4}{2}\right)}=\frac{3\pi^{\frac{L}{2}}}{4\Gamma\left(\frac{L+4}{2}\right)}.
\end{align}
The fourth moment is
\begin{align}
    S^{(4)}=&\frac{4}{\mathcal{V}^2}\times\frac{3^2\pi^L}{4^2\Gamma^2\left(\frac{L+4}{2}\right)}\sum_{l_1,l_2=1}^{L}T_{l_1l_2}^4.
\end{align}
		
(ii) Two of the elements are equal to $2$, and all other elements are zero. There are $3$ cases.
		
(a) $\alpha_{l_1}=4$ and $\tilde{\alpha}_{l_2}=\tilde{\alpha}_{l_3}=2$.
\begin{align}
    \alpha_{l_1}&=4,~
    \tilde{\alpha}_{l_2}=2,~\tilde{\alpha}_{l_3}=2,\nonumber\\
    \beta_{l_1}&=\frac{5}{2},~
    \tilde{\beta}_{l_2}=\frac{3}{2},~
    \tilde{\beta}_{l_2}=\frac{3}{2},\nonumber\\
    B(\beta_{1},\cdots,\beta_{L})&=B\left(\frac{1}{2},\cdots,\frac{5}{2},\cdots,\frac{1}{2}\right)\nonumber\\
    &=\frac{\left[\Gamma\left(\frac{1}{2}\right))\right]^{L-1}\Gamma\left(\frac{5}{2}\right)}{\Gamma\left(\frac{L+4}{2}\right)}
    =\frac{3\pi^{\frac{L}{2}}}{4\Gamma\left(\frac{L+4}{2}\right)},\nonumber\\
    B(\tilde{\beta}_{1},\cdots,\tilde{\beta}_{L})&=B\left(\frac{1}{2},\cdots,\frac{3}{2},\cdots,\frac{3}{2},\cdots,\frac{1}{2}\right)\nonumber\\
    &=\frac{\left[\Gamma\left(\frac{1}{2}\right))\right]^{L-2}\left[\Gamma\left(\frac{3}{2}\right)\right]^{2}}{\Gamma\left(\frac{L+4}{2}\right)}\nonumber\\
    &=\frac{\pi^{\frac{L}{2}}}{4\Gamma\left(\frac{L+4}{2}\right)}.
\end{align}
We find
\begin{align}
    S^{(4)}=&\frac{4}{\mathcal{V}^2}\times6\times\frac{1}{2}\times\frac{3\pi^{L}}{4^2\Gamma^2\left(\frac{L+4}{2}\right)}\sum_{\substack{l_1,l_2,l_3=1\\l_2\neq l_3}}^{L}T_{l_1l_2}^2T_{l_1l_3}^2.
\end{align}
		
(b) $\alpha_{l_1}=\alpha_{l_2}=2$ and $\tilde{\alpha}_{l_3}=4$.
\begin{align}
    \alpha_{l_1}&=2,~
    \alpha_{l_2}=2,~\tilde{\alpha}_{l_3}=4,~\nonumber\\
    \beta_{l_1}&=\frac{3}{2},~
    \beta_{l_2}=\frac{3}{2},~
    \tilde{\beta}_{l_2}=\frac{5}{2},\nonumber\\
    B(\beta_{1},\cdots,\beta_{L})&=B\left(\frac{1}{2},\cdots,\frac{3}{2},\cdots,\frac{3}{2},\cdots,\frac{1}{2}\right)\nonumber\\
    &=\frac{\left[\Gamma\left(\frac{1}{2}\right))\right]^{L-2}\left[\Gamma\left(\frac{3}{2}\right)\right]^{2}}{\Gamma\left(\frac{L+4}{2}\right)}\nonumber\\
    &=\frac{\pi^{\frac{L}{2}}}{4\Gamma\left(\frac{L+4}{2}\right)},\nonumber\\
    B(\tilde{\beta}_{1},\cdots,\tilde{\beta}_{L})&=B\left(\frac{1}{2},\cdots,\frac{5}{2},\cdots,\frac{1}{2}\right)\nonumber\\
    &=\frac{\left[\Gamma\left(\frac{1}{2}\right))\right]^{L-1}\Gamma\left(\frac{5}{2}\right)}{\Gamma\left(\frac{L+4}{2}\right)}\nonumber\\
    &=\frac{3\pi^{\frac{L}{2}}}{4\Gamma\left(\frac{L+4}{2}\right)}.
\end{align}
We find a similar result
\begin{align}
    S^{(4)}=&\frac{4}{\mathcal{V}^2}\times6\times\frac{1}{2}\times\frac{3\pi^{L}}{4^2\Gamma^2\left(\frac{L+4}{2}\right)}\sum_{\substack{l_1,l_2,l_3=1\\l_1\neq l_2}}^{L}T_{l_1l_3}^2T_{l_2l_3}^2.
\end{align}
		
(c) $\alpha_{l_1}=\alpha_{l_2}=2$ and $\tilde{\alpha}_{l_3}=\tilde{\alpha}_{l_4}=2$.
\begin{align}
    \alpha_{l_1}&=2,~
    \alpha_{l_2}=2,~
    \tilde{\alpha}_{l_3}=2,~\tilde{\alpha}_{l_4}=2,~\nonumber\\
    \beta_{l_1}&=\beta_{l_2}=\frac{3}{2},~
    \tilde{\beta}_{l_3}=\tilde{\beta}_{l_4}=\frac{3}{2},\nonumber\\
    B(\beta_{1},\cdots,\beta_{L})&=B(\tilde{\beta}_{1},\cdots,\tilde{\beta}_{L})\nonumber\\
    &=B\left(\frac{1}{2},\cdots,\frac{3}{2},\cdots,\frac{3}{2},\cdots,\frac{1}{2}\right)\nonumber\\
    &=\frac{\left[\Gamma\left(\frac{1}{2}\right))\right]^{L-2}\left[\Gamma\left(\frac{3}{2}\right)\right]^{2}}{\Gamma\left(\frac{L+4}{2}\right)}\nonumber\\
    &=\frac{\pi^{\frac{L}{2}}}{4\Gamma\left(\frac{L+4}{2}\right)}.
\end{align}
We find
\begin{align}
    S^{(4)}=&\frac{4}{\mathcal{V}^2}\times6\times\frac{1}{2}\times\frac{\pi^{L}}{4^2\Gamma^2\left(\frac{L+4}{2}\right)}\sum_{\substack{l_1,l_2,l_3,l_4=1\\l_1\neq l_3,l_2\neq l_4}}^{L}T_{l_1l_2}^2T_{l_3l_4}^2.
\end{align}
		
(iii) For $t_{l_1,l_2}=t_{l_3,l_4}=t_{l_5,l_6}=t_{l_7,l_8}=1$, only one case yields finite values of the integrals, $t_{l_1,l_2}=t_{l_1,l_3}=1$, $t_{l_4,l_2}=t_{l_4,l_3}=1$. In this case, we have
\begin{align}
    \alpha_{l_1}&=t_{l_1,l_2}+t_{l_1,l_3}=2,\nonumber\\
    \alpha_{l_2}&=t_{l_4,l_2}+t_{l_4,l_3}=2,\nonumber\\
    \tilde{\alpha}_{l_2}&=t_{l_1,l_2}+t_{l_4,l_2}=2,\nonumber\\
    \tilde{\alpha}_{l_3}&=t_{l_1,l_3}+t_{l_4,l_3}=2,\nonumber\\
    B(\beta_1,\cdots,\beta_{L})&=B\left(\frac{1}{2},\cdots,\frac{3}{2},\cdots,\frac{3}{2},\cdots,\frac{1}{2}\right)\nonumber\\
    &=\frac{\pi^{\frac{L}{2}}}{4\Gamma\left(\frac{L+4}{2}\right)},\nonumber\\
    B(\tilde{\beta}_1,\cdots,\tilde{\beta}_{L})&=B\left(\frac{1}{2},\cdots,\frac{3}{2},\cdots,\frac{3}{2},\cdots,\frac{1}{2}\right)\nonumber\\
    &=\frac{\pi^{\frac{L}{2}}}{4\Gamma\left(\frac{L+4}{2}\right)}.
\end{align}
Therefore, the fourth moment turns to
\begin{align}
    S^{(4)}=&\frac{4}{\mathcal{V}^2}\times24\times\frac{1}{4}\times\frac{\pi^{L}}{4^2\Gamma^2\left(\frac{L+4}{2}\right)}\nonumber\\
    &\times\sum_{\substack{l_1,l_2,l_3,l_4=1\\l_2\neq l_3,l_1\neq l_4}}^{L}T_{l_1l_2}T_{l_1l_3}T_{l_4l_2}T_{l_4l_3}.
\end{align}
		
In summary, we conclude
\begin{align}
    S^{(4)}=W&\Big[3\sum_{l_1,l_2=1}^{L}T_{l_1l_2}^4
    +3\sum_{\substack{l_1,l_2,l_3=1\\l_2\neq l_3}}^{L}T_{l_1l_2}^2T_{l_1l_3}^2\Big.\nonumber\\
    &\Big.+3\sum_{\substack{l_1,l_2,l_3=1\\l_1\neq l_2}}^{L}T_{l_1l_3}^2T_{l_2l_3}^2
    +\sum_{\substack{l_1,l_2,l_3,l_4=1\\l_1\neq l_3,l_2\neq l_4}}^{L}T_{l_1l_2}^2T_{l_3l_4}^2+\Big.\nonumber\\
    &\Big.2\sum_{\substack{l_1,l_2,l_3,l_4=1\\l_2\neq l_3,l_1\neq l_4}}^{L}T_{l_1l_2}T_{l_1l_3}T_{l_4l_2}T_{l_4l_3}\Big],
\end{align}
where the constant
\begin{align}
    W&=\frac{4}{\mathcal{V}^2}\times\frac{3\pi^L}{4^2\Gamma^2\Big(\frac{L+4}{2}\Big)}
    =\frac{1}{\mathcal{V}^2}\frac{3\pi^L}{4\Gamma^2\Big(\frac{L+4}{2}\Big)}\nonumber\\
    &=\frac{\Gamma^2(\frac{L}{2})}{4\pi^L}\frac{3\pi^L}{4\Gamma^2\Big(\frac{L+4}{2}\Big)}
    =\frac{3\Gamma^2(\frac{L}{2})}{16\Gamma^2\Big(\frac{L+4}{2}\Big)}.
\end{align}
		
Extract the submatrix $T_{\mathrm{R}}=(T_{j,k})$ with $j,k\in\{1,2,\cdots,L\}$ and denote its singular values $\tau_j$, $j=1,\cdots,L$. The second and fourth moments can be expressed as
\begin{align}
    S^{(2)}&=V\sum_{j=1}^{L}\tau_j^2,\\
    S^{(4)}&=W\Big[2\sum_{j=1}^{L}\tau_j^4+\Big(\sum_{j=1}^{L}\tau_j^2\Big)^2\Big]\nonumber\\
    &=W\Big[2\sum_{j=1}^{L}\tau_j^4+L^4\left[S^{(2)}\right]^2\Big].
\end{align}
Note that $\sum_{j=1}^{L}\tau_j^4=\sum_{j,k=1}^{L}T_{jk}^4$ and $\sum_{j=1}^{L}\tau_j^2=\sum_{j,k=1}^{L}T_{jk}^2$.
		
Now, we construct a suitable observable such that $Q^{(2)}=S^{(2)}$. Suppose the diagonal and traceless observable $M_A=M_B=\mathrm{diag}\left(k_1,k_2,\cdots,k_d\right)$ with real numbers $k_j$ and $\sum_{j=1}^{d}k_j=0$, $j=1,2,\cdots,d$. After calculating $Q^{(2)}$, we find that suitable numbers $k_j$ should satisfy the condition $\sum_{j=1}^{d}k_j^2=d$. The proof is the same as the calculation of $Q^{(2)}$ in Appendix~\ref{Moments}. For example, $M_A=M_B=\mathrm{diag}\left(\sqrt{3/2},0,-\sqrt{3/2}\right)$ for $d=3$.
		
It is possible to construct suitable observables such that $Q^{(4)}=S^{(4)}$. For different $d$, we numerical text that the following observables constructed in Ref.~\cite{wyderka2023probing} are suitable, 
\begin{align}
    &d\!=\!3,M\!=\!\mathrm{diag}\left(\sqrt{3/2},0,-\sqrt{3/2}\right),\nonumber\\
    &d\!=\!4,M\!=\!\mathrm{diag}\left(1.357,0.400,-0.400,-1.357\right),\nonumber\\
    &d\!=\!5,M\!=\!\mathrm{diag}\left(1.444, 0.644,0,-0.644,-1.444\right),\nonumber\\
    &d\!=\!6,M\!=\!\mathrm{diag}\left(1.4966,0.8719,-1.4966, -0.8719,0,0\right),\nonumber\\
    &d\!=\!7,M\!=\!\mathrm{diag}\left(1.5041,1.1125,-1.5041,-1.1125,0,0,0\right),
\end{align}
where $M_A=M_B=M$.
		
\section{Other numerical results and analytic lower bound}\label{D_Numerical}
Here we provide several examples for evaluating our entanglement detection criterion presented in the main text.
		
For two-qutrit states, we consider Bell diagonal states \cite{moerland2024bound}. Given the generalized Bell state $|\phi^{00}\rangle=\frac{1}{\sqrt{d}}\sum_{j=0}^{d-1}|jj\rangle$, we define a basis of a bipartite Hilbert space $|\phi^{\alpha\beta}\rangle=\left(Z^{\alpha}\otimes X^{\beta}\right)|\phi^{00}\rangle$ for parameters $\alpha,\beta=0,1,\cdots,d-1$. The shift operator $X:|j\rangle\leftarrow|j\oplus1\rangle$ and the clock operator $Z:|j\rangle\leftarrow\omega^{j}|j\rangle$, where $\oplus$ denotes addition modulo $d$ and $\omega=e^{2\pi i/d}$. For a given probability distribution $\{p_{\alpha\beta}|p_{\alpha\beta}\geq0,\sum_{\alpha,\beta}p_{\alpha\beta}=1\}$, we define the generalized Bell diagonal state
\begin{align}
    \vr_{P}=\sum_{\alpha,\beta=0}^{d-1}p_{\alpha\beta}|\phi^{\alpha\beta}\rangle\langle\phi^{\alpha\beta}|.
\end{align}
Fig.~\ref{SFig1} (a) indicates that $\vr_P$ with the following probability $\vr_P$ is detectable from RRMs,
\begin{align}
    P_1&\!=\!\begin{pmatrix}
		0 & 0 & 0 \\
		0 & 0.18 & 0.82 \\
		0 & 0 & 0
		\end{pmatrix},
    P_2\!=\!\begin{pmatrix}
		0.52 & 0 & 0 \\
		0 & 0 & 0 \\
		0 & 0.48 & 0
		\end{pmatrix},\nonumber\\
    P_3&\!=\!\begin{pmatrix}
        0.5 & 0 & 0 \\
        0 & 0 & 0 \\
        0 & 0 & 0.5
        \end{pmatrix},
    P_4\!=\!\begin{pmatrix}
        0.045 & 0 & 0 \\
        0 & 0.2055 & 0 \\
        0 & 0 & 0.7495
        \end{pmatrix}.
\end{align}
		
For the $4\otimes4$ bound state, the following Piani state cannot be detected by RRMs \cite{piani2007class}. Consider the orthogonal projections $P_{jk}=|\Phi_{jk}\rangle\langle\Phi_{jk}|$, where $|\Phi_{jk}\rangle=\left(\eins\otimes\sigma_{j}\otimes\sigma_k\right)\sum_{k=0}^{3}\frac{1}{2}|kk\rangle$ with Pauli operators $\sigma_j$, $\sigma_k$. With these projections, the Piani state is defined as
\begin{align}
    \vr=\frac{1}{6}\left(P_{02}+P_{11}+P_{23}+P_{31}+P_{32}+P_{33}\right).
\end{align}
Fig.~\ref{SFig1} (b) indicates that the Piani state cannot be detected.
		
\begin{figure*}[ht]
    \centering
    \includegraphics[scale=0.5]{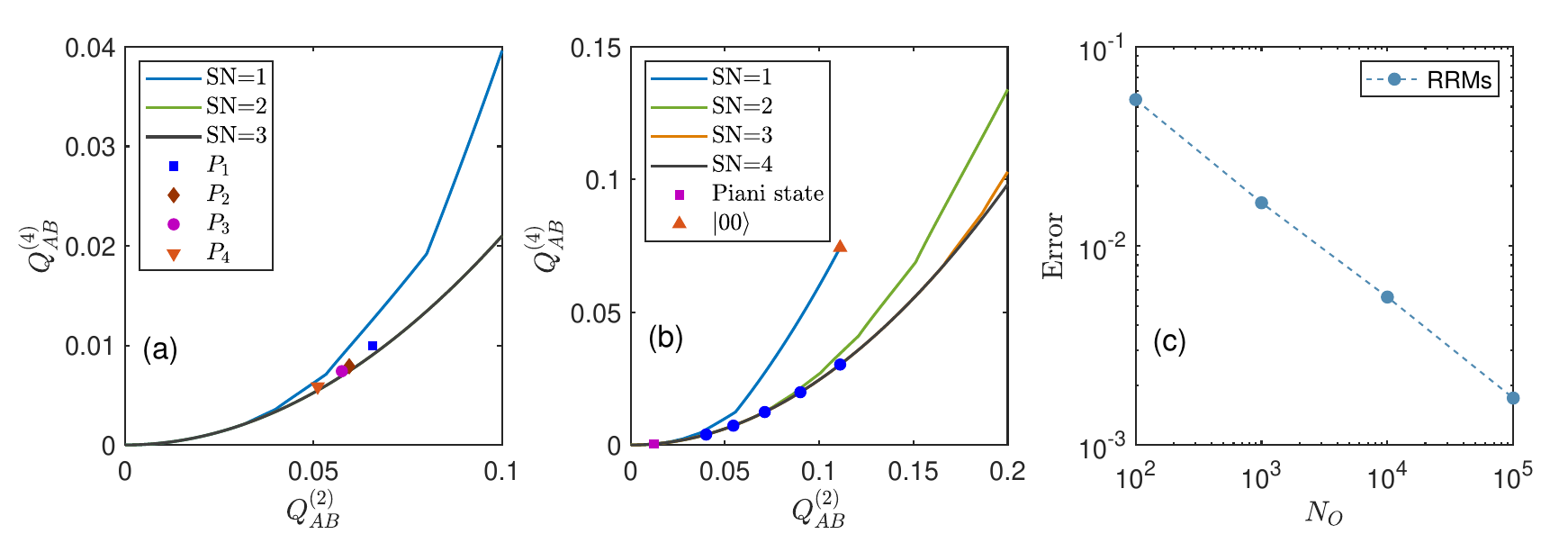}
    \caption{Entanglement dimensionality criterion based on second and fourth moments of RRMs for (a) two-qutrit and (b) $4\otimes4$ systems. (a) The entangled states ($P_1,P_2,P_3,P_4$) are outside the separable area, meaning that our criterion can detect these states. (b) shows our criterion can detect different states. The blue dots from right to left denote the $4\otimes4$ isotropic state $\vr_{\mathrm{iso}}$ with $p=1,0.9,0.8,0.7,0.6$. (c) The average error for estimating the overlap $\tr(\vr_1\vr_2)$ of two $5\otimes5$ states $\vr_1=\vr(0.1)$ and $\vr_2=\vr(0.9)$.}
    \label{SFig1}
\end{figure*}

\begin{figure}[ht]
    \centering
    \includegraphics[scale=0.5]{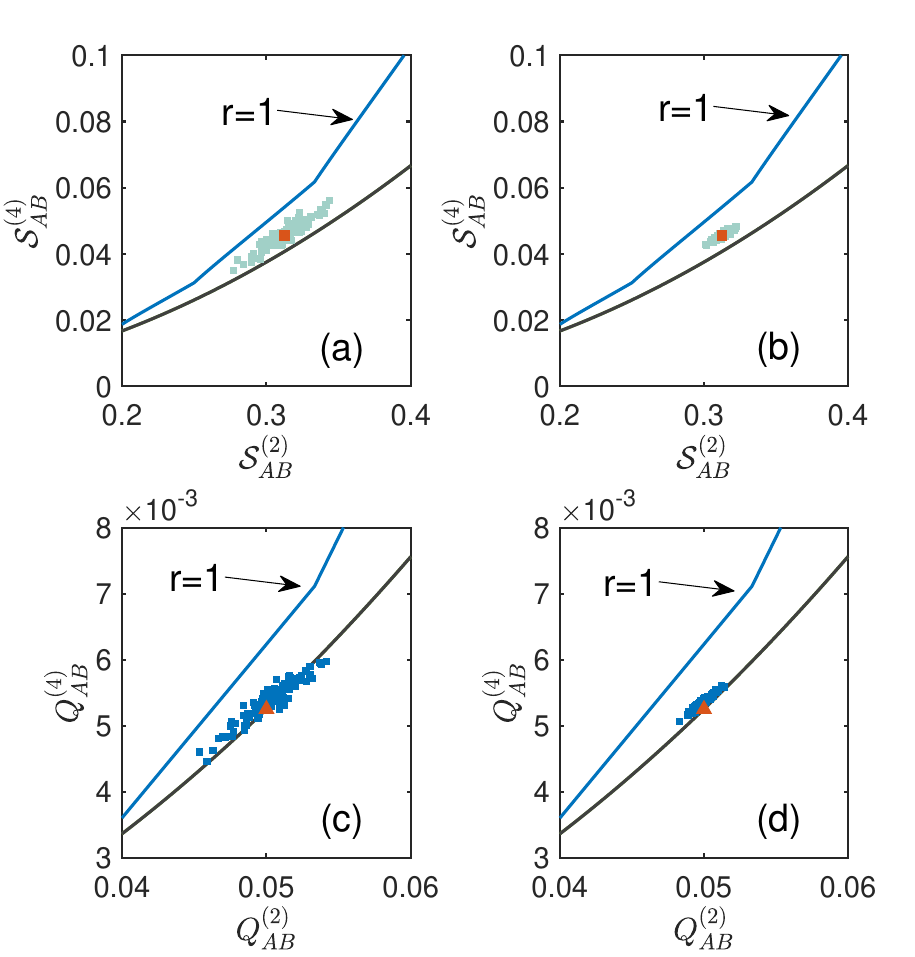}
    \caption{Comparison between RMs and RRMs for detecting the chess board state $\vr_{cbs}$. (a) and (b) show $100$ simulation points by using $10^3$ and $10^4$ random unitaries in RMs, respectively. (c) and (d) show $100$ simulation points by using $10^3$ and $10^4$ random orthogonal matrices in RRMs, respectively. The triangle and square denote the exact point for $\vr_{cbs}$.}
    \label{SFig2}
\end{figure}

We present details about Fig. 3 in the main text. For Fig. 3, we repeat the classical shadow procedure as shown in observation 4 $N_{\mathrm{iter}}=100$ times and obtain $N_{\mathrm{iter}}$ errors $E(i)=|f_{\mathrm{es}}(i)-f_{\mathrm{ex}}|$, $i=1,\cdots,N_{\mathrm{iter}}$. In Fig. 3, we plot the average of the error $\sum_{i=1}^{N_{\mathrm{iter}}}E(i)/N_{\mathrm{iter}}$ for a $5$-qubit $\vr_p$ with classical shadows obtained from RMs and RRMs.

Finally, we compare the sample number of unitary or orthogonal matrices for detecting a specific state. Fig.~\ref{SFig2} shows the comparison between RMs and RRMs for the chess board states $\vr_{cbs}$. The results indicate that RMs and RRMs can detect the state $\vr_{cbs}$ with the same number of random operations.

\section{Generalized results}\label{E_Generalized}
In this section, we present several generalized results, including the detection of entanglement and imaginarity, the estimation of the overlap, and the classical shadow.
		
\subsection{Entanglement criterion of $N$-qudit state}\label{Generalized:entanglement}
For multipartite systems, denoting an $N$-qudit $\vr=\frac{1}{d^N}\sum_{j_1,\cdots,j_N=0}^{d^2-1}T_{j_1\cdots j_N}\lambda_{j_1}\otimes\cdots\otimes\lambda_{j_N}$, the second moment of RRMs is $Q^{(2)}=L^{-N}\sum_{j_1,\cdots,j_N=1}^{L}T_{j_1\cdots j_N}^2$. Based on the second moment, we present an entanglement criterion explained by the following observation.
\begin{observation}\label{Ob:EnDetectN}
Any separable $N$-qudit state $\vr$ obeys $Q^{(2)}\leq\left[\frac{2}{d+2}\right]^N$.
\end{observation}
\begin{proof}
For all separable $N$-qudit states, the sector length satisfies $\mathcal{S}_N=\sum_{j_1,\cdots,j_N=1}^{d^2-1}T_{j_1\cdots j_N}^2\leq(d-1)^N$ \cite{aschauer2004local,cieslinski2024analysing}. The second moment of RRMs has an upper bound
\begin{align}
    Q^{(2)}&=\left[\frac{2}{(d-1)(d+2)}\right]^N\sum_{j_1,\cdots,j_N=1}^{L}T_{j_1\cdots j_N}^2\\
    &\leq \left[\frac{2}{(d-1)(d+2)}\right]^N\sum_{j_1,\cdots,j_N=1}^{d^2-1}T_{j_1\cdots j_N}^2\\
    &=\left[\frac{2}{d+2}\right]^N.
\end{align}
\end{proof}
		
The violation of this bound (observation~\ref{Ob:EnDetectN}) implies that the state is entangled.

\subsection{Imaginarity detection and quantification of multi-qudit states}\label{Generalized:imaginarity}
In the main text, we present the imaginarity detection and quantification of two-qudit states. Next, we first consider the single-qudit states.
\begin{observation}\label{Ob:ImagSingle}
    Let $\vr_s$ be a single-qudit state and denote its second moments of RMs and RRMs as $R_{s}^{(2)}$ and $Q_{s}^{(2)}$. We define a gap,
    \begin{align}
        G_{s}=(d^2-1)R_{s}^{(2)}-LQ_{s}^{(2)}.
    \end{align}
    Then we have (a) $\vr_s$ is an imaginary state if $G_{s}>0$. Otherwise, $\vr_s$ is a real state. (b) $G_{s}$ induces a lower bound of the robustness of the imaginarity measure $\mathcal{F}_{R}(\vr_{AB})=\min_{\tau}\left\{\mu\geq0:\frac{\vr_{AB}+\mu\tau}{1+\mu}\in\mathscr{R}\right\}$~\cite{wu2021operational}, such that $\sqrt{{G_{s}}/{d}}\leq \mathcal{F}_{R}(\vr_s)$, where the equality holds if and only if $\vr_s$ is a real state and $\mathscr{R}$ denotes the set of all real states.
\end{observation}
\begin{proof}
    First of all, we establish a connection between the gap $G_{s}$ and the Hilbert-Schmidt norm $\|\vr_s-\vr_s^{\top}\|_{\mathrm{HS}}^2$, where $\vr_s^{\top}$ denotes the transpose of the state $\vr_s$. Divide $\vr_s$ into two parts associated with real GGMs and imaginary GGMs, $\vr_s=\frac{1}{d}\eins_d+\frac{1}{d}\sum_{j=1}^{L}T_j\lambda_j+\frac{1}{d}\sum_{k=L+1}^{L+\hat{L}}T_k\lambda_k$. We define a matrix
    \begin{align}
    \vr_s-\vr_s^{\top}
    &=\frac{1}{d}\sum_{j=1}^{L}T_j\lambda_j
    +\frac{1}{d}\sum_{k=L+1}^{L+\hat{L}}T_k\lambda_k\nonumber\\
    &-\frac{1}{d}\sum_{j=1}^{L}T_j\lambda_j^{\top}
    -\frac{1}{d}\sum_{k=L+1}^{L+\hat{L}}T_k\lambda_k^{\top}\nonumber\\
    &=\frac{2}{d}\sum_{k=L+1}^{L+\hat{L}}T_k\lambda_k,
    \end{align}
    The square of the Hilbert-Schmidt norm of matrix $\vr_s-\vr_s^{\top}$ is
    \begin{align}
        \|\vr_s-\vr_s^{\top}\|_{\mathrm{HS}}^2
        &=\left\|\frac{2}{d}\sum_{j=L+1}^{L+\hat{L}}T_j\lambda_j\right\|_{\mathrm{HS}}^2\nonumber\\
        &=\tr\left[\frac{4}{d^2}\sum_{j,k=L+1}^{L+\hat{L}}T_jT_k\lambda_j\lambda_k\right]\nonumber\\
        &=\frac{4}{d}\sum_{j=L+1}^{L+\hat{L}}T_j^2.
    \end{align}
    The gap is defined as
    \begin{align}
        G_{s}&=(d^2-1)R_{s}^{(2)}-LQ_{s}^{(2)}\nonumber\\
        &=\sum_{j=1}^{d^2-1}T_{j}^2-\sum_{k=1}^{L}T_{k}^2\nonumber\\
        &=\sum_{j=L+1}^{L+\hat{L}}T_{j}^2\nonumber\\
        &=\hat{L}\hat{Q}_s^{(2)}\geq0.
    \end{align}
    Therefore, $G_{s}=\frac{d}{4}\|\vr_s-\vr_s^{\top}\|_{\mathrm{HS}}^2$. It is easy to check that $G_{s}=0$ if and only if $\vr_s=\vr_s^{\top}$, which is a real state. Otherwise, $G_{s}>0$ indicates $\vr$ is an imaginary state. We complete the proof of result (a).
			
    For result (b), the inequality $\|A\|_{\tr}\geq\|A\|_{\mathrm{HS}}$ holds, and the equality is true if and only if the operator $A$ has all zero singular values. Thus, the robustness of the imaginarity measure has a lower bound
    \begin{align}
        \mathcal{F}_{R}(\vr_s)=\frac{1}{2}\|\vr_s-\vr_s^{\top}\|_{\tr}&\geq\frac{1}{2}\|\vr_s-\vr_s^{\top}\|_{\mathrm{HS}}\nonumber\\
        &=\frac{1}{2}\sqrt{\frac{4G_{s}}{d}}=\sqrt{\frac{G_{s}}{d}}.
    \end{align}
    The equality holds if and only if $\vr_s-\vr_s^{\top}$ has zero singular values and, equivalently, $\vr_s$ is a real state.
\end{proof}
		
The next observation shows that if we only utilize the moments of RMs and RRMs, we can obtain a sufficient condition for detecting the imaginarity of two-qudit states.
\begin{observation}\label{Ob:ImagTwo1}
    Let $\vr_{AB}$ be a two-qudit quantum state. Then $\vr_{AB}$ is an imaginary state if one of the following conditions is true, (a) $G_A>0$, (b) $G_B>0$, where $G_A$ and $G_B$ denote the gap of reduced states $\vr_{A}$, $\vr_{B}$ defined in observation~\ref{Ob:ImagSingle}.
\end{observation}
		
\begin{proof}
    Consider a two-qudit quantum state
    \begin{align}
        \vr_{AB}=\frac{1}{d^2}\sum_{j,k=0}^{d^2-1}T_{jk}\lambda_j^A\otimes\lambda_k^B.
    \end{align}
    Two reduced states are
    \begin{align}
        \vr_A&=\frac{1}{d}\sum_{j=0}^{d^2-1}T_{j0}\lambda_j^A,~
        \vr_B=\frac{1}{d}\sum_{k=0}^{d^2-1}T_{0k}\lambda_k^B.
    \end{align}
    If one of the reduced states $\vr_s$ is an imaginary state, then state $\vr_{AB}$ is also an imaginary state, where $s\in\{A,B\}$. $\vr_s$ is an imaginary state equivalent to the corresponding quantity $G_s>0$. Thus, if $G_s>0$, then $\vr$ is an imaginary state.
\end{proof}
We remark that $\vr_{AB}$ may be an imaginary state even if both conditions (a,b) in observation~\ref{Ob:ImagTwo1} are not true. The correlation of $\vr_{AB}$ contains imaginarity, which also implies $\vr_{AB}$ is imaginary. Observation~\ref{Ob:ImagN} is a direct generalization of observation~\ref{Ob:ImagTwo1} to $N$-qudit states. 
\begin{observation}\label{Ob:ImagN}
    Let $\rho$ be an $N$-qudit quantum state and partition the whole system into $m$ nonempty disjoint subsystems $A_j$, $j=1,\cdots,m$. Then $\vr$ is an imaginary state if the second moment of the reduced state $\vr_{s}=\tr_{\bar{s}}(\vr)$, $G_s>0$ for a fixed set $s$ and $\bar{s}=\{A_1,\cdots,A_{\alpha-1},A_{\alpha},\cdots,A_m\}$.
\end{observation}
The proof is similar to Observation~\ref{Ob:ImagTwo1}. Observation~\ref{Ob:ImagN} is also a sufficient condition.
		
\subsection{Classical shadow with RRMs}\label{Generalized:classical}
The classical shadow is a method for predicting properties of quantum systems based on RMs \cite{huang2020predicting}. One applies a random local or global unitary $U$ on an $N$-qudit quantum state $\vr$ and then performs projective measurement in the computational basis $\{|b\rangle\}$ with outcome bit string $b=(b_1,\cdots,b_N)$ and $b_n\in\{0,1,\cdots,d-1\}$ for $n=1,\cdots,N$. The above procedure is repeated $N_U$ independently unitaries $U$. One can use the obtained data to extract interesting information about quantum systems. The transformation $U$ is picked from the Clifford group or the Pauli group for $d=2$~\cite{elben2023randomized}. We here find that RRMs with global orthogonal matrices are sufficient for analyzing real states.
		
The classical shadow with RRMs employs a global orthogonal matrix $O$ to rotate the state, $\vr\rightarrow O\vr O^{\top}$. After measuring the updated state in the basis $\{|b\rangle\}$, we store a classical snapshot $O^{\top}|b\rangle\langle b|O$ whose expectation value over the random orthogonal matrix and probability distribution can be linked with the measurement channel
\begin{align}
    \mathcal{M}(\rho)=\sum_{b}\int dO\langle b|O\vr O^{\top}|b\rangle O^{\top}|b\rangle\langle b|O.
\end{align}
Assuming that $\mathcal{M}$ is invertible, $\tilde{\vr}:=\mathcal{M}^{-1}(O^{\top}|b\rangle\langle b|O)$ is the classical shadow of $\vr$ and $\mathbb{E}(\tilde{\vr})=\vr$. Below, we observe that the measurement channel $\mathcal{M}(\cdot)$ is related to the second moment of RRMs.

Therefore, the classical shadow for the $m$th measurement setting $O^{(m)}$ with $K$ outcomes $|b^{(m)}\rangle=|b_1^{(m)}\cdots b_N^{(m)}\rangle$, is
\begin{align}
    \tilde{\vr}^{(m)}:=\frac{d^N+2}{2}O^{(m)\top}|b^{(m)}\rangle\langle b^{(m)}|O^{(m)}-\frac{1}{2}\eins_{d^{N}}.
\end{align}
The collection $\{\tilde{\vr}^{(m)}\}_{m=1\cdots,N_O}$ is called a classical shadow of $\vr$. We can further predict properties such as (i) the expectation value of observable $\mathcal{G}$, $g=\tr(\mathcal{G}\tilde{\vr})$, (ii) quadratic functions of the quantum state.
		
Here, two questions arise. First, how about the upper bounds for the sample complexity of the above tasks? However, this problem involves the $3$-design of the orthogonal group. We will leave this question for future work. Second, does it still work for local orthogonal matrices in the classical shadow? The answer is case by case. The proof of this claim follows. Suppose we perform a local orthogonal matrix $O=\bigotimes_{j=1}^{N}O_j$. For elementary tensor products $X=X_1\otimes\cdots\otimes X_N$ and each $X_j$ is a $d\times d$ matrix, we obtain the measurement channel is
\begin{widetext}
    \begin{align}
    \mathcal{M}(X)&=\sum_{b}\int dO\langle b|O(X_1\otimes\cdots\otimes X_N)O^{\top}|b\rangle O^{\top}|b\rangle\langle b|O\nonumber\\
    &=\bigotimes_{j=1}^{N}\sum_{b_j\in\{0,\cdots,d-1\}}\int dO_j\langle b_j|O_jX_jO_j^{\top}|b_j\rangle O_j^{\top}|b_j\rangle\langle b_j|O_j\nonumber\\
    &=\bigotimes_{j=1}^{N}\frac{1}{d(d+2)}\sum_{b_j\in\{0,\cdots,d-1\}}\tr_1\left[X_j\otimes\eins_d+(X_j\otimes\eins_d)\mathbb{S}+(X_j\otimes\eins_d)\Pi\right]\nonumber\\
    &=\bigotimes_{j=1}^{N}\frac{\tr(X_j)\eins_d+X_j+X_j^{\top}}{d+2}.
\end{align}
\end{widetext}
It is clear to check that if $X_j=X_j^{\top}$ for all $j$, the above equation turns to
\begin{align}
    \mathcal{M}(X)&=\bigotimes_{j=1}^{N}\frac{\tr(X_j)\eins_d+2X_j}{d+2}.
\end{align}
Thus, the classical shadows of $X$ are deduced via the inverse of the channel,
\begin{align}
    \tilde{X}&=\mathcal{M}^{-1}(O^{\top}|b\rangle\langle b|O)\nonumber\\
    &=\bigotimes_{j=1}^{N}\frac{d+2}{2}O_j^{\top}|b_j\rangle\langle b_j|O_j-\frac{1}{2}\eins_{d}.
\end{align}
Extending linearly to all quantum states in $\mathcal{L}(\mathds{C}^{d^N})$, we find that the type of quantum states that are implementable for local orthogonal matrices are the PTI states. Here, we define the PTI states as $\vr=\vr^{\top_j}$ for any $j$.
		
\subsection{Cross-platform verification of quantum devices}\label{Cross-platformV}
We next present a protocol with RRMs and PRRMs for cross-platform verification of quantum computers \cite{elben2020cross}. In particular, we will estimate the fidelity $\mathcal{F}(\vr_1,\vr_2)=\frac{\tr(\vr_1\vr_2)}{\max\{\tr(\vr_1^2),\tr(\vr_2^2)\}}$ of two-qudit real states $\vr_1$, $\vr_2$ to verify whether two quantum devices have prepared the same quantum state~\cite{liang2019quantum}. Following (i-iii) in {Assumptions}, we present a strategy to estimate the overlap $\tr(\vr_1\vr_2)$ as well as the purities of states.
		
\subsubsection{Estimation of the overlap between bipartite states from local orthogonal matrices}
From Assumptions (i) and (ii), we first apply both $\vr_1$, $\vr_2$ to the same random local orthogonal matrix, $O=O_A\otimes O_B$, where $O_A$ and $O_B$ are sampled via the Haar measure on orthogonal group $\boldsymbol{\mathrm{O}}(d)$. Consider two $d^2\times d^2$ real local observables:
\begin{align}
    &M_1=\left(\alpha_1|0\rangle\langle0|+\alpha_2|d-1\rangle\langle d-1|\right)^{\otimes 2},\nonumber\\
    &M_2=\left(\beta_1|0\rangle\langle0|+\beta_2|d-1\rangle\langle d-1|\right)^{\otimes 2},
\end{align}
where real numbers $\alpha_1$, $\alpha_2$, $\beta_1$, and $\beta_2$ to satisfy the following conditions:
\begin{align}
    (d+1)(\alpha_1+\alpha_2)(\beta_1+\beta_2)&=2(\alpha_1\beta_1+\alpha_2\beta_2).\nonumber\\
    d(\alpha_1\beta_1+\alpha_2\beta_2)&\neq(\alpha_1+\alpha_2)(\beta_1+\beta_2).
\end{align}
We estimate two expectation values $E_k=\tr\left(O\vr_kO^{\top}M_k\right)$ for $k=1,2$. Repeating the above procedure for different orthogonal matrices $O$, we find the quantity, $\mathcal{D}=\int dOE_1E_2$.
		
Similarly, we consider two imaginary observables
\begin{align}
    \hat{M}_1=\hat{M}_2=(-i|0\rangle\langle d-1|+i|d-1\rangle\langle0|)^{\otimes2},
\end{align}
and find another quantity $\hat{\mathcal{D}}=\int dO\hat{E}_1\hat{E}_2$, where expectation values $\hat{E}_k=\tr(O\vr_kO^{\top}\hat{M}_1)$ for $k=1,2$. Using $\mathcal{D}$ and $\hat{\mathcal{D}}$, we present the following result:
\begin{observation}\label{Ob:OverlapEstimation}
    Consider any two-qudit state $\vr_1$ and a real two-qudit state $\vr_2$ and apply RRMs and PRRMs. The overlap is estimated by
    \begin{align}
        \tr(\vr_1\vr_2)=\frac{\eta^2\mathcal{D}+\gamma^2\hat{\mathcal{D}}}{4\gamma^2\eta^2},
    \end{align}
    where parameters $\gamma=\frac{-(\alpha_1+\alpha_2)(\beta_1+\beta_2)+d(\alpha_1\beta_1+\alpha_2\beta_2)}{d\left(d-1\right)\left(d+2\right)}$ and $\eta=\frac{2}{d(d-1)}$.
\end{observation}
\begin{proof}
    Consider two-qudit states $\vr_1$, $\vr_2$, we can estimate the overlap $\tr(\vr_1\vr_2)$ by RRMs and PRRMs with local orthogonal matrices. Construct real observables for real parameters $\alpha_1,\alpha_2,\beta_1,\beta_2$ (determined later),
    \begin{align}
        M_1&=M_1^A\otimes M_1^B\nonumber\\
        &=\left(\alpha_1|0\rangle\langle0|+\alpha_2|d-1\rangle\langle d-1|\right)^{\otimes 2},\nonumber\\
        M_2&=M_2^A\otimes M_2^B\nonumber\\
        &=\left(\beta_1|0\rangle\langle0|+\beta_2|d-1\rangle\langle d-1|\right)^{\otimes 2}.
    \end{align}
    
    \onecolumngrid
    We have
    \begin{align}
        \tr(M_1^k\otimes M_2^k)&=(\alpha_1+\alpha_2)(\beta_1+\beta_2),\nonumber\\
        \tr(\mathbb{S}M_1^k\otimes M_2^k)&=\tr(M_1^kM_2^k)=\alpha_1\beta_1+\alpha_2\beta_2,\nonumber\\
        \tr(\Pi M_1^k\otimes M_2^k)&=\tr\left(\mathbb{S}\hat{M}_1^k\otimes\left(\hat{M}_2^k\right)^{\top}\right)
        =\tr(\mathbb{S}M_1^k\otimes M_2^k)
        =\alpha_1\beta_1+\alpha_2\beta_2,
    \end{align}
    for $k=A,B$. In this scenario, the integral is
    \begin{align}
        \mathcal{D}&=\int dO\tr\left(O\vr_1O^{\top}M_1\right)\tr\left(O\vr_2O^{\top}M_2\right)\nonumber\\
        &=\tr\left[\int dO(O^{\top}\otimes O^{\top})(M_1\otimes M_2)(O\otimes O)(\vr_1\otimes\vr_2)\right]\nonumber\\
        &=\tr\left[\int dO_A\int dO_B(O_A^{\top}\otimes O_B^{\top}\otimes O_A^{\top}\otimes O_B^{\top})(M_1^A\otimes M_1^B\otimes M_2^A\otimes M_2^B)(O_A\otimes O_B\otimes O_A\otimes O_B)(\vr_1\otimes\vr_2)\right]\nonumber\\
        &=\tr\left[\bigotimes_{k=A,B}\int dO_k(O_k^{\top}\otimes O_k^{\top})(M_1^k\otimes M_2^k)(O_k\otimes O_k)(\vr_1\otimes\vr_2)\right].
    \end{align}
    Using Lemma~\ref{Lemma1}, we obtain the following quantity,
    \begin{align}
        \int dO_k(O_k^{\top}\otimes O_k^{\top})(M_1^k\otimes M_2^k)(O_k\otimes O_k)&=c_1\eins_{d^2}+c_2\mathbb{S}_k+c_3\Pi_k,
    \end{align}
    where the coefficients are
    \begin{align}
        c_1&=\frac{(d+1)\tr(M_1^k\otimes M_2^k)-\tr(\mathbb{S}M_1^k\otimes M_2^k)-\tr(\Pi M_1^k\otimes M_2^k)}{d(d-1)(d+2)}\nonumber\\
        &=\frac{(d+1)(\alpha_1+\alpha_2)(\beta_1+\beta_2)-2(\alpha_1\beta_1+\alpha_2\beta_2)}{d(d-1)(d+2)},\nonumber\\
        c_2&=\frac{-\tr(M_1^k\otimes M_2^k)+(d+1)\tr(\mathbb{S}M_1^k\otimes M_2^k)-\tr(\Pi M_1^k\otimes M_2^k)}{d(d-1)(d+2)}\nonumber\\
        &=\frac{-(\alpha_1+\alpha_2)(\beta_1+\beta_2)+d(\alpha_1\beta_1+\alpha_2\beta_2)}{d(d-1)(d+2)},\nonumber\\
        c_3&=\frac{-\tr(M_1^k\otimes M_2^k)-\tr(\mathbb{S}M_1^k\otimes M_2^k)+(d+1)\tr(\Pi M_1^k\otimes M_2^k)}{d(d-1)(d+2)}\nonumber\\
        &=c_2.
    \end{align}
    Hence, the quantity $\mathcal{D}$ is
    \begin{align}
		\mathcal{D}&=\tr\left[\bigotimes_{k=A,B}\left(c_1\eins_{d^2}+c_2\mathbb{S}_k+c_2\Pi_k\right)(\vr_1\otimes\vr_2)\right]\nonumber\\
        &=\gamma^{2}\tr\left[\bigotimes_{k=A,B}\left(\mathbb{S}_k+\Pi_k\right)(\vr_1\otimes\vr_2)\right]\nonumber\\
        &=\gamma^{2}\tr\left[\left(\mathbb{S}_A\otimes\mathbb{S}_B+\mathbb{S}_A\otimes\Pi_B+\Pi_A\otimes\mathbb{S}_B+\Pi_A\otimes\Pi_B\right)(\vr_1\otimes\vr_2)\right],
    \end{align}
    To remove the tensor item that includes single operators $\mathbb{S}_k$ and $\Pi_k$, we let $c_1=0$ and parameters $\alpha_1,\alpha_2,\beta_1,\beta_2$ should satisfy
    \begin{align}
        (d+1)(\alpha_1+\alpha_2)(\beta_1+\beta_2)=2(\alpha_1\beta_1+\alpha_2\beta_2).
    \end{align}
    The parameter $\gamma=c_2$. From the calculation, we can conclude that the quantity $\mathcal{D}$ is meaningful only when the parameter $\gamma$ is nonzero. Thus the free parameters $\alpha_j$ and $\beta_j$, $j=1,2$ have another constraints such that
    \begin{align}
        d(\alpha_1\beta_1+\alpha_2\beta_2)\neq(\alpha_1+\alpha_2)(\beta_1+\beta_2).
    \end{align}
			
    Now we define two two-qudit imaginary observables,
    \begin{align}
        &\hat{M}_1=\hat{M}_1^A\otimes \hat{M}_1^B=(-i|0\rangle\langle d-1|+i|d-1\rangle\langle0|)^{\otimes2},\nonumber\\
        &\hat{M}_2=\hat{M}_2^A\otimes \hat{M}_2^B=(-i|0\rangle\langle d-1|+i|d-1\rangle\langle0|)^{\otimes2}.
    \end{align}
    We have
    \begin{align}
        \tr(\hat{M}_1^k\otimes \hat{M}_2^k)&=0,\nonumber\\
        \tr(\mathbb{S}\hat{M}_1^k\otimes \hat{M}_2^k)&=\tr(\hat{M}_1^k\hat{M}_2^k)=2,\nonumber\\
        \tr(\Pi \hat{M}_1^k\otimes \hat{M}_2^k)&=\tr\left(\mathbb{S}\hat{M}_1^k\otimes\left(\hat{M}_2^k\right)^{\top}\right)
        =-\tr(\mathbb{S}\hat{M}_1^k\otimes \hat{M}_2^k)=-2.
    \end{align}
    In this scenario, the integral is
    \begin{align}
        \hat{\mathcal{D}}&=\int dO\tr\left(O\vr_1O^{\top}\hat{M}_1\right)\tr\left(O\vr_2O^{\top}\hat{M}_2\right)\nonumber\\
        &=\tr\left[\bigotimes_{k=A,B}\int dO_k(O_k^{\top}\otimes O_k^{\top})(\hat{M}_1^k\otimes \hat{M}_2^k)(O_k\otimes O_k)(\vr_1\otimes\vr_2)\right].
    \end{align}
    Using Lemma~\ref{Lemma1}, we obtain the following quantity,
    \begin{align}
        \int dO_k(O_k^{\top}\otimes O_k^{\top})(\hat{M}_1^k\otimes \hat{M}_2^k)(O_k\otimes O_k)&=\hat{c}_1\eins_{d^2}+\hat{c}_2\mathbb{S}_k+\hat{c}_3\Pi_k,
    \end{align}
    where the coefficients are
    \begin{align}
        \hat{c}_1&=\frac{(d+1)\tr(\hat{M}_1^k\otimes \hat{M}_2^k)-\tr(\mathbb{S}\hat{M}_1^k\otimes \hat{M}_2^k)-\tr(\Pi \hat{M}_1^k\otimes \hat{M}_2^k)}{d(d-1)(d+2)}=0,\nonumber\\
        \hat{c}_2&=\frac{-\tr(\hat{M}_1^k\otimes \hat{M}_2^k)+(d+1)\tr(\mathbb{S}\hat{M}_1^k\otimes \hat{M}_2^k)-\tr(\Pi \hat{M}_1^k\otimes \hat{M}_2^k)}{d(d-1)(d+2)}=\frac{2}{d(d-1)},\nonumber\\
        \hat{c}_3&=\frac{-\tr(\hat{M}_1^k\otimes \hat{M}_2^k)-\tr(\mathbb{S}\hat{M}_1^k\otimes \hat{M}_2^k)+(d+1)\tr(\Pi \hat{M}_1^k\otimes \hat{M}_2^k)}{d(d-1)(d+2)}
        =\frac{-2}{d(d-1)}=-\hat{c}_2.
    \end{align}
    The integral turns to
    \begin{align}
        \hat{\mathcal{D}}&=\eta^{2}\tr\left[\bigotimes_{k=A,B}\left(\mathbb{S}_k-\Pi_k\right)(\vr_1\otimes\vr_2)\right]\nonumber\\
        &=\eta^{2}\tr\left[\left(\mathbb{S}_A\otimes\mathbb{S}_B-\mathbb{S}_A\otimes\Pi_B-\Pi_A\otimes\mathbb{S}_B+\Pi_A\otimes\Pi_B\right)(\vr_1\otimes\vr_2)\right],
    \end{align}
    where the parameter $\eta=\frac{2}{d(d-1)}$.
			
    Add the above two integrals with suitable coefficients, and we have
    \begin{align} \nonumber
        \frac{\mathcal{D}}{\gamma^2}+\frac{\hat{\mathcal{D}}}{\eta^2}
        &=2\tr\left[\left(\mathbb{S}_A\otimes\mathbb{S}_B+\Pi_A\otimes\Pi_B\right)(\vr_1\otimes\vr_2)\right]\\ \label{eq:balbalatest}
        &=4\tr\left[\left(\mathbb{S}_A\otimes\mathbb{S}_B\right)(\vr_1\otimes\vr_2)\right]\nonumber\\
        &=4\tr\left[\vr_1\vr_2\right].
    \end{align}
    In Eq.(\ref{eq:balbalatest}), we use the fact that for any state $\vr_1$ and the real state $\vr_2$, $\vr_2=\vr^{\top}$, the following equation is true
    \begin{align}
        \tr\left[\left(\Pi_A\otimes\Pi_B\right)\left(\vr_1\otimes\vr_2\right)\right]
        &=\tr\left[\left(\mathbb{S}_A^{\top_2}\otimes\mathbb{S}_B^{\top_2}\right)\left(\vr_1\otimes\vr_2\right)\right]\nonumber\\
        &=\tr\left[\left(\mathbb{S}_A^{\top_2}\otimes\mathbb{S}_B^{\top_2}\right)\left(\sum_{\boldsymbol{l},\boldsymbol{m},\boldsymbol{j},\boldsymbol{k}}\vr_1^{\boldsymbol{l}\boldsymbol{m}}\vr_2^{\boldsymbol{j}\boldsymbol{k}}|l_A\rangle\langle\tilde{l}_A|\otimes|m_B\rangle\langle\tilde{m}_B|\otimes|j_A\rangle\langle\tilde{j}_A|\otimes|k_B\rangle\langle\tilde{k}_B|\right)\right]\nonumber\\
        &=\tr\left[\left(\mathbb{S}_A\otimes\mathbb{S}_B\right)\left(\sum_{\boldsymbol{l},\boldsymbol{m},\boldsymbol{j},\boldsymbol{k}}\vr_1^{\boldsymbol{l}\boldsymbol{m}}\vr_2^{\boldsymbol{j}\boldsymbol{k}}|l_A\rangle\langle\tilde{l}_A|\otimes|m_B\rangle\langle \tilde{m}_B|\otimes|j_A\rangle\langle\tilde{j}_A|^{\top}\otimes|k_B\rangle\langle\tilde{k}_B|^{\top}\right)\right]\nonumber\\
        &=\tr\left[\left(\mathbb{S}_A\otimes\mathbb{S}_B\right)\left(\sum_{\boldsymbol{l},\boldsymbol{m},\boldsymbol{j},\boldsymbol{k}}\vr_1^{\boldsymbol{l}\boldsymbol{m}}\vr_2^{\boldsymbol{j}\boldsymbol{k}}|l_A\rangle\langle\tilde{l}_A|\otimes|m_B\rangle\langle \tilde{m}_B|\otimes|j_A\rangle\langle\tilde{j}_A|\otimes|k_B\rangle\langle\tilde{k}_B|\right)\right]\nonumber\\
        &=\tr\left[\left(\mathbb{S}_A\otimes\mathbb{S}_B\right)\left(\vr_1\otimes\vr_2\right)\right].
    \end{align}
    We here expand $\vr_1$ and $\vr_2$ in the computational basis  of subsystems A and B, such that
    \begin{align}\label{TwoStates}
        \vr_1=\sum_{\boldsymbol{l},\boldsymbol{m}}\vr_1^{\boldsymbol{l}\boldsymbol{m}}|l_A\rangle\langle \tilde{l}_A|\otimes|m_B\rangle\langle\tilde{m}_B|,
        \vr_2=\sum_{\boldsymbol{j},\boldsymbol{k}}\vr_2^{\boldsymbol{j}\boldsymbol{k}}|j_A\rangle\langle \tilde{j}_A|\otimes|k_B\rangle\langle\tilde{k}_B|.
    \end{align}
    For the real state $\vr_2$, we have
    \begin{align}
        \vr_2^{\top}=\sum_{\boldsymbol{j},\boldsymbol{k}}\vr_2^{\boldsymbol{j}\boldsymbol{k}}|j_A\rangle\langle \tilde{j}_A|^{\top}\otimes|k_B\rangle\langle\tilde{k}_B|^{\top}
        =\vr_2=\sum_{\boldsymbol{j},\boldsymbol{k}}\vr_2^{\boldsymbol{j}\boldsymbol{k}}|j_A\rangle\langle \tilde{j}_A|\otimes|k_B\rangle\langle\tilde{k}_B|.
    \end{align}
    Finally, we obtain the overlap
    \begin{align}
        \tr\left[\vr_1\vr_2\right]=\frac{\eta^2\mathcal{D}+\gamma^2\hat{\mathcal{D}}}{4\gamma^2\eta^2}.
    \end{align}
    \end{proof}
    
The proof indicates that our protocol can estimate the overlap $\tr(\vr_1\vr_2)$ where $\vr_2$ must be a real state and $\vr_1$ can be any quantum state.
		
\subsubsection{Estimation of the overlap between multipartite states from global and local orthogonal matrices}\label{Generalized:estimation}
We present protocols to estimate the overlap $\tr(\vr_1\vr_2)$ of two $N$-qudit states $\vr_1$ and $\vr_2$ by RRMs and PRRMs with global and local orthogonal matrices.
		
\textbf{Estimation from global orthogonal matrix}
		
Given two $N$-qudit states $\vr_1$ and $\vr_2$, we first apply to both $\vr_1$, $\vr_2$ the same random global orthogonal matrix, $O$. Here $O$ is sampled via the Haar measure on the orthogonal group $\boldsymbol{\mathrm{O}}(d)$. We calculate the expectation value on both states with two $d^N\times d^N$ diagonal observables
\begin{align}
    &M_1=\alpha_1|0\rangle\langle0|^{\otimes N}+\alpha_2|d-1\rangle\langle d-1|^{\otimes N},
    M_2=\beta_1|0\rangle\langle0|^{\otimes N}+\beta_2|d-1\rangle\langle d-1|^{\otimes N},
\end{align}
where $\{|j\rangle\}$ is the computational basis in a $d$-dimension space and nonzero real numbers $\alpha_1$, $\alpha_2$, $\beta_1$, and $\beta_2$ satisfy
\begin{align}
    \left(d^N+1\right)(\alpha_1+\alpha_2)(\beta_1+\beta_2)=2(\alpha_1\beta_1+\alpha_2\beta_2),
\end{align}
which enables the second moment operator to be only spanned by the operators $\mathbb{S}$ and $\Pi$. We then find that
\begin{align}
    \tr(M_1\otimes M_2)=(\alpha_1+\alpha_2)(\beta_1+\beta_2),
    \tr(\mathbb{S}M_1\otimes M_2)=\tr(\Pi M_1\otimes M_2)=\alpha_1\beta_1+\alpha_2\beta_2.
\end{align}
Hence, the integral turns to
\begin{align}
    \int dO\tr\left(O\vr_1O^{\top}M_1\right)\tr\left(O\vr_2O^{\top}M_2\right)
    &=\int dO\tr\left[\left(O\otimes O\right)\left(\vr_1\otimes \vr_2\right)\left(O^{\top}\otimes O^{\top}\right)\left(M_1\otimes M_2\right)\right]\nonumber\\
    &=\int dO\tr\left[\left(O^{\top}\otimes O^{\top}\right)\left(M_1\otimes M_2\right)\left(O\otimes O\right)\left(\vr_1\otimes \vr_2\right)\right]\nonumber\\
    &=\tr\left[\int dO\left(O^{\top}\otimes O^{\top}\right)\left(M_1\otimes M_2\right)\left(O\otimes O\right)\left(\vr_1\otimes \vr_2\right)\right].
\end{align}
Using Lemma~\ref{Lemma1}, we find
    \begin{align}
        \int dO(O^{\top}\otimes O^{\top})(M_1\otimes M_2)(O\otimes O)&=c_1\eins_{d^N}+c_2\mathbb{S}+c_3\Pi,
    \end{align}
where the coefficients are
    \begin{align}
        c_1&=\frac{(d^N+1)\tr(M_1\otimes M_2)-\tr(\mathbb{S}M_1\otimes M_2)-\tr(\Pi M_1\otimes M_2)}{d^N(d^N-1)(d^N+2)}\nonumber\\
        &=\frac{(d^N+1)(\alpha_1+\alpha_2)(\beta_1+\beta_2)-2(\alpha_1\beta_1+\alpha_2\beta_2)}{d^N(d^N-1)(d^N+2)},\nonumber\\
        c_2&=\frac{-\tr(M_1\otimes M_2)+(d^N+1)\tr(\mathbb{S}M_1\otimes M_2)-\tr(\Pi M_1\otimes M_2)}{d^N(d^N-1)(d^N+2)}\nonumber\\
        &=\frac{-(\alpha_1+\alpha_2)(\beta_1+\beta_2)+d^N(\alpha_1\beta_1+\alpha_2\beta_2)}{d^N(d^N-1)(d^N+2)},\nonumber\\
        c_3&=\frac{-\tr(M_1\otimes M_2)-\tr(\mathbb{S}M_1\otimes M_2)+(d^N+1)\tr(\Pi M_1\otimes M_2)}{d^N(d^N-1)(d^N+2)}\nonumber\\
        &=c_2.
    \end{align}
Let $c_1=0$ such that $(d^N+1)(\alpha_1+\alpha_2)(\beta_1+\beta_2)=2(\alpha_1\beta_1+\alpha_2\beta_2)$. The integral is
\begin{align}
    \int dO(O^{\top}\otimes O^{\top})(M_1\otimes M_2)(O\otimes O)&=\gamma\left(\mathbb{S}+\Pi\right),~\gamma=c_2.
\end{align}
Now, we have
\begin{align}
    \int dO\tr\left(O\vr_1O^{\top}M_1\right)\tr\left(O\vr_2O^{\top}M_2\right)
    &=\tr\left[\gamma\left(\mathbb{S}+\Pi\right)\left(\vr_1\otimes \vr_2\right)\right]
    =2\gamma\tr\left[\mathbb{S}\left(\vr_1\otimes \vr_2\right)\right]
    =2\gamma\tr\left(\vr_1\vr_2\right).
\end{align}
The second equation holds since
\begin{align}
    \tr\left[\Pi\left(\vr_1\otimes\vr_2\right)\right]&=\tr\left[\mathbb{S}\left(\vr_1\otimes\vr_2^{\top}\right)\right]=\tr\left[\mathbb{S}\left(\vr_1\otimes\vr_2\right)\right]
\end{align}
for the real state $\vr_2=\vr_2^{\top}$. Thus, the overlap is estimated by
\begin{align}
    \tr(\vr_1\vr_2)=\frac{d^N\left(d^N-1\right)\left(d^N+2\right)}{-2(\alpha_1+\alpha_2)(\beta_1+\beta_2)+2d^N(\alpha_1\beta_1+\alpha_2\beta_2)}\int dO\tr\left(O\vr_1O^{\top}M_1\right)\tr\left(O\vr_2O^{\top}M_2\right).
\end{align}
In summary, the estimation from a global orthogonal matrix requires one of the quantum states to be a real state.
		
\textbf{Estimation from local orthogonal matrices}
Suppose we first apply to both $N$-qudit states $\vr_1$, $\vr_2$ the same random local orthogonal matrix, $O=\bigotimes_{k=1}^{N}O_k$. We construct real local observables
\begin{align}\label{LocalObservables1}
    M_1=\bigotimes_{k=1}^{N}M_1^{k},M_2=\bigotimes_{k=1}^{N}M_2^{k},
\end{align}
where $d\times d$ real observables are
\begin{align}
    M_1^{k}=\alpha_1|0\rangle\langle0|+\alpha_2|d-1\rangle\langle d-1|,
    M_2^{k}=\beta_1|0\rangle\langle0|+\beta_2|d-1\rangle\langle d-1|.
\end{align}
Thus, the integral is
\begin{align}\label{LocalOrthogonal}
    \int dO\tr\left(O\vr_1O^{\top}M_1\right)\tr\left(O\vr_2O^{\top}M_2\right)
    &=\tr\left[\int dO(O^{\top}\otimes O^{\top})(M_1\otimes M_2)(O\otimes O)(\vr_1\otimes\vr_2)\right]\nonumber\\
    &=\tr\left[\bigotimes_{k=1}^{N}\int_{dO_k}(O_k^{\top}\otimes O_k^{\top})(M_1^{k}\otimes M_2^{k})(O_k\otimes O_k)(\vr_1\otimes\vr_2)\right]\nonumber\\
    &=\gamma^{N}\tr\left[\bigotimes_{k=1}^{N}\left(\mathbb{S}_k+\Pi_k\right)(\vr_1\otimes\vr_2)\right]\nonumber\\
    &=(2\gamma)^{N}\tr\left[\bigotimes_{k=1}^{N}\mathbb{S}_k(\vr_1\otimes\vr_2)\right]\nonumber\\
    &=(2\gamma)^{N}\tr\left[\mathbb{S}(\vr_1\otimes\vr_2)\right]\nonumber\\
    &=(2\gamma)^{N}\tr\left[\vr_1\vr_2\right].
\end{align}
The fourth equation is true if and only if $\vr_2$ are partial transpose invariant (PTI) such that $\vr_2=\vr_2^{\top}=\vr_2^{\top_k}$ for $k=1,\cdots,N$. In this case, the overlap is
\begin{align}
    \tr\left[\vr_1\vr_2\right]=\left[\frac{1}{2\gamma}\right]^{N}\int dO\tr\left(O\vr_1O^{\top}M_1\right)\tr\left(O\vr_2O^{\top}M_2\right).
\end{align}
		
In conclusion, the estimation from the local orthogonal matrix requires one of the quantum states to be a PTI state. For example, $3\otimes3$ bound states such as chessboard and UPB states are PTI states.
		
Here, we prove the fourth equality of Eq.~(\ref{LocalOrthogonal}) for $N=2$ and rewrite two systems as $A,B$. The proof for general $N$ is similar to $N=2$. Using Eq.~(\ref{TwoStates}), we have
\begin{align}
    \tr\left[(\mathbb{S}_A+\Pi_A)\otimes(\mathbb{S}_B+\Pi_B)(\vr_1\otimes\vr_2)\right]=\tr\left[\left(\mathbb{S}_A\otimes\mathbb{S}_B+\mathbb{S}_A\otimes\Pi_B+\Pi_A\otimes\mathbb{S}_B+\Pi_A\otimes\Pi_B\right)(\vr_1\otimes\vr_2)\right].
\end{align}
Without loss of generality, we calculate
\begin{align}
    \tr\left[\left(\mathbb{S}_A\otimes\Pi_B\right)\left(\vr_1\otimes\vr_2\right)\right]
    &=\tr\left[\left(\mathbb{S}_A\otimes\mathbb{S}_B^{\top_2}\right)\left(\vr_1\otimes\vr_2\right)\right]\nonumber\\
    &=\tr\left[\left(\mathbb{S}_A\otimes\mathbb{S}_B^{\top_2}\right)\left(\sum_{\boldsymbol{l},\boldsymbol{m},\boldsymbol{j},\boldsymbol{k}}\vr_1^{\boldsymbol{l}\boldsymbol{m}}\vr_2^{\boldsymbol{j}\boldsymbol{k}}|l_A\rangle\langle\tilde{l}_A|\otimes|m_B\rangle\langle\tilde{m}_B|\otimes|j_A\rangle\langle\tilde{j}_A|\otimes|k_B\rangle\langle\tilde{k}_B|\right)\right]\nonumber\\
    &=\tr\left[\left(\mathbb{S}_A\otimes\mathbb{S}_B\right)\left(\sum_{\boldsymbol{l},\boldsymbol{m},\boldsymbol{j},\boldsymbol{k}}\vr_1^{\boldsymbol{l}\boldsymbol{m}}\vr_2^{\boldsymbol{j}\boldsymbol{k}}|l_A\rangle\langle\tilde{l}_A|\otimes|m_B\rangle\langle \tilde{m}_B|\otimes|j_A\rangle\langle\tilde{j}_A|\otimes|k_B\rangle\langle\tilde{k}_B|^{\top}\right)\right]\nonumber\\
    &=\tr\left[\left(\mathbb{S}_A\otimes\mathbb{S}_B\right)\left(\sum_{\boldsymbol{l},\boldsymbol{m},\boldsymbol{j},\boldsymbol{k}}\vr_1^{\boldsymbol{l}\boldsymbol{m}}\vr_2^{\boldsymbol{j}\boldsymbol{k}}|l_A\rangle\langle\tilde{l}_A|\otimes|m_B\rangle\langle \tilde{m}_B|\otimes|j_A\rangle\langle\tilde{j}_A|\otimes|k_B\rangle\langle\tilde{k}_B|\right)\right]\nonumber\\
    &=\tr\left[\left(\mathbb{S}_A\otimes\mathbb{S}_B\right)\left(\vr_1\otimes\vr_2\right)\right],
\end{align}
where the fourth equality holds if $\vr_2=\vr_2^{\top_B}$. Similarly, $\tr\left[\left(\Pi_A\otimes\mathbb{S}_B\right)\left(\vr_1\otimes\vr_2\right)\right]=\tr\left[\left(\mathbb{S}_A\otimes\mathbb{S}_B\right)\left(\vr_1\otimes\vr_2\right)\right]$ means $\vr_2=\vr_2^{\top_A}$ and $\tr\left[\left(\Pi_A\otimes\Pi_B\right)\left(\vr_1\otimes\vr_2\right)\right]=\tr\left[\left(\mathbb{S}_A\otimes\mathbb{S}_B\right)\left(\vr_1\otimes\vr_2\right)\right]$ means $\vr_2=\vr_2^{\top}$. In summary, the following equality is true
\begin{align}
    \tr\left[(\mathbb{S}_A+\Pi_A)\otimes(\mathbb{S}_B+\Pi_B)(\vr_1\otimes\vr_2)\right]=4\tr\left[\left(\mathbb{S}_A\otimes\mathbb{S}_B\right)(\vr_1\otimes\vr_2)\right],
\end{align}
if and only if $\vr_2$ is a PTI state.
		
To test our protocol, we estimate the overlap $\tr(\vr_1\vr_2)$ of two $5\otimes5$ states $\vr_1=\vr(0.1)$ and $\vr_2=\vr(0.9)$, where the parameterized quantum state $\vr(p)=p|\psi\rangle\langle\psi|+(1-p)\frac{\eins_5\otimes\eins_5}{5^2}$, and the pure state $|\psi\rangle=\frac{1}{\sqrt{5}}\sum_{i=0}^{4}|ii\rangle$. We utilize the approach in observation~\ref{Ob:OverlapEstimation} and obtain an overlap estimator $f_{\mathrm{es}}$, and then find an estimated error $E=|\tr(\vr_1\vr_2)-f_{\mathrm{es}}|$. Repeating the above procedure $N_{\mathrm{iter}}$ times, we obtain $N_{\mathrm{iter}}$ estimated errors $E(i)$, $i=1,\cdots,N_{\mathrm{iter}}$. Fig.~\ref{SFig1} (c) presents the average error $\sum_{i=1}^{N_{\mathrm{rm}}}E(i)/N_{\mathrm{iter}}$ with $N_{\mathrm{iter}}=500$ and different number of random orthogonal matrices. The parameters in observables are set as $\alpha_1=0$, $\alpha_2=1$, $\beta_1=1$, and $\beta_2=-\frac{3}{2}$. The results indicate that $10^4$ orthogonal can obtain an overlap estimation up to an error $5.5\times10^{-3}$.

\twocolumngrid
\bibliography{RRMs}
\end{document}